\documentclass[a4]{article}
\usepackage[utf8]{inputenc}
\usepackage{fullpage}
\usepackage{amsmath, amssymb}
\usepackage{graphicx}
\usepackage[nocompress]{cite}
\usepackage{enumerate}
\usepackage{comment}
\usepackage{appendix}

\usepackage[dvipsnames,usenames]{xcolor}
\RequirePackage{hyperref}
\hypersetup{colorlinks=true, urlcolor=Blue, citecolor=Green, linkcolor=BrickRed, breaklinks, unicode}

\newtheorem{theorem}{Theorem}

\newtheorem{lemma}[theorem]{Lemma}
\newtheorem{corollary}[theorem]{Corollary}

\newcommand{\rmq}{\ensuremath\mathsf{rmq}}
\newcommand{\update}{\ensuremath\mathsf{update}}
\newcommand{\weight}{\ensuremath\mathsf{weight}}
\newcommand{\lca}{\ensuremath\mathsf{lca}}

\newcommand{\pmini}{\ensuremath\mathsf{pmin}}
\newcommand{\mini}{\ensuremath\mathsf{min}}
\newcommand{\compress}{\ensuremath\mathsf{compress}}
\newcommand{\spread}{\ensuremath\mathsf{spread}}

\newcommand{\ceil}[1]{\left\lceil{#1}\right\rceil}
\newcommand{\floor}[1]{\left\lfloor{#1}\right\rfloor}
\renewcommand{\angle}[1]{\langle{#1}\rangle}

\usepackage{tikz}
\usetikzlibrary{arrows, quotes, patterns}
\tikzstyle{Red Dot}=[fill=red, draw=black, shape=circle]
\tikzstyle{Green Dot}=[fill=white, draw=green, shape=circle]
\tikzstyle{Empty}=[fill=white, draw=black, shape=circle]
\tikzstyle{Black}=[fill=black, draw=black, shape=circle]
\tikzstyle{Red fill}=[-, fill={rgb,255: red,255; green,88; blue,91}]
\tikzstyle{grey fill}=[fill={rgb,255: red,207; green,207; blue,207}, draw=black, shape=circle]
\tikzstyle{Red Circle}=[fill=white, draw=red, shape=circle]
\tikzstyle{letter}=[fill=none, draw=none, shape=circle]

\tikzstyle{Left}=[<-]
\tikzstyle{Right}=[->]

\pgfdeclarelayer{nodelayer}
\pgfdeclarelayer{edgelayer}
\pgfsetlayers{edgelayer,nodelayer,main}

\begin{document}
\title{Dynamic Range Minimum Queries on the Ultra-Wide Word RAM\footnote{An extended abstract appeared at the \emph{50th International Conference on Current Trends in Theory and Practice of Computer Science}~\cite{BGSL2024}.}}
%
%
\author{Philip Bille \\ \texttt{phbi@dtu.dk} \and Inge Li G{\o}rtz \\ \texttt{inge@dtu.dk} \\ M{\'a}ximo P{\'e}rez-L{\'o}pez \\\texttt{mpelo@dtu.dk} \and Tord Stordalen}

%
%
%
\maketitle              
\begin{abstract}
We consider the dynamic range minimum problem on the ultra-wide word RAM model of computation. This model extends the classic $w$-bit word RAM model with special ultrawords of length $w^2$ bits that support standard arithmetic and boolean operation and scattered memory access operations that can access $w$ (non-contiguous) locations in memory. The ultra-wide word RAM model captures (and idealizes) modern vector processor architectures. The goal in the dynamic range minimum problem is to maintain an array $A$ of $n$ $w$-bit integers subject to range minimum queries (given indices $i$ and $j$ return a smallest integer in the subarray $A[i..j]$) and updates (given index $i$ and integer $\alpha$ set $A[i] \leftarrow \alpha$).

Our main result is a data structure that supports range minimum queries and updates in $O(\log \log \log n)$ time and uses $O(n/\log n)$ space in addition to the input array. This exponentially improves the time of existing techniques. Our result is based on a simple reduction to prefix minimum computations on sequences $O(\log n)$ words combined with a new parallel, recursive implementation of these.
\end{abstract}

\section{Introduction} 
Supporting \emph{range minimum queries} (RMQ) on arrays is a well-studied, classic data structure problem, see e.g.,~\cite{harel1984fast, Yao1985, CR1989, GBT1984, BBGSV1989, AGKR2004, FH2011, BFC2000, DLW2014, YA2010, BDLRR2016, ST1983, PD2006, BCR1996, AHR1998, BDR2011}. This paper considers the \emph{dymamic RMQ problem} defined as follows: maintain an array $A[0, \ldots, n-1]$ of $w$-bit integers and subject to the following operations. 
\begin{itemize}
	\item $\rmq(i,j)$: return a smallest integer in the subarray $A[i..j]$. 
	\item $\update(i, \alpha)$: set $A[i] \leftarrow \alpha$. 
\end{itemize}
On most models of computation, the complexity of the dynamic RMQ problem is well-understood~\cite{ST1983, PD2006, BCR1996, AHR1998, BDR2011}. For instance, on the word RAM, a tight $\Theta(\log n/\log \log n)$ time bound on the operations is known~\cite{BDR2011}. Hence, a natural question is whether practical models of computation capturing modern hardware advances will allow us to improve this bound significantly. One such model is the \emph{ultra-wide word RAM model} (UWRAM) introduced by Farzan et al.~\cite{FLNS2015}. The UWRAM extends the word RAM model with special \emph{ultrawords} of $w^2$ bits. The model supports standard boolean and arithmetic operations on ultrawords and \emph{scattered} memory operations that access $w$ words in parallel. The UWRAM model captures (and idealizes) modern vector processing architectures~\cite{CRDI2007, LNOM2008, Reinders2013, Stephensetal2017}. We present the details of the UWRAM model of computation in Section~\ref{sec:uwram_model}. By extending recent techniques for the UWRAM model~\cite{BGS2020}, we can immediately solve the dynamic RMQ problem using $O(\log \log n)$ time per operation.  

\subsection{Results and Techniques}
Our main result is an exponential improvement of the $O(\log \log n)$ time bound. More precisely, we show the following bound:  

\begin{theorem}\label{thm:main}
	Given an array $A$ of $n$ $w$-bit integers, we can construct in $O(n)$ time a data structure on the UWRAM that supports $\rmq$ and $\update$ in $O(\log \log \log n)$ time and uses $O(n/\log n)$ words of space in addition to $A$.
\end{theorem}

Technically, our solution is based on a simple linear space and logarithmic time folklore solution, which we call the \emph{range minimum tree} (see Section~\ref{sec:rangeminimumtree} for a detailed description). The range minimum tree is a balanced binary tree over the input array and supports operations in $O(\log n)$ time by sequentially traversing the tree. On the UWRAM, we show how to efficiently compute the access patterns of the operations on the range minimum tree in parallel using scattered memory access operations and prefix minimum computation on ultrawords. More precisely, we show that a given algorithm for a prefix minimum computation on a sequence of words of length $\ell = O(w)$ stored in a constant number of ultrawords that uses time $t(\ell)$ implies a linear space solution for dynamic RMQ that supports both operations in $O(t(\log n))$ time. If we implement a standard parallel prefix computation algorithm~\cite{LF1980} (see also the survey by Blelloch~\cite{Blelloch1990}) using the UWRAM techniques in Bille et al.~\cite{BGS2022} we immediately obtain an algorithm that uses $O(\log \ell)$ time. Our main technical contribution is a new, exponentially faster prefix minimum algorithm that achieves $O(\log \log \ell)$ time. The key idea is a constant time prefix minimum algorithm for "short" sequences of $O(\sqrt{w})$ words that takes advantage of parallel computations on multiple copies of the sequence packed into a constant number of ultrawords (note that a constant number of ultrawords can store $O(\sqrt{w})$ copies of a sequence of $O(\sqrt{w})$ words). We implement the idea recursively and in parallel on each recursion level to obtain a fast solution for general sequences of words of length $O(w)$. Each recursion step uses constant time, and the depth is $O(\log \log \ell)$, leading to the $O(\log \log \ell)$ time bound. By combining an implicit heap layout with a standard indirection technique we show how to implement the data structure using $O(n/\log n)$ words of space in addition to the $n$ words used by the input array. 

\subsection{Outline}
The paper is organized as follows. In Section~\ref{sec:uwram_model} and~\ref{sec:rangeminimumtree} we review the UWRAM model of
computation and the range minimum tree. In Section~\ref{sec:uwramrangeminimumtree}, we present the UWRAM implementation of the range minimum tree that leads to the reduction to prefix minimum computation on word sequences. In Sections~\ref{subsec:reducing_space} and~\ref{subsec:heap} we show how to implement the reduction using in $O(n)$ space and $O(n/\log n)$ space, respectively, in addition to the input array. Finally, in Section~\ref{sec:prefixminwordsequence}, we present our fast prefix minimum algorithm and plug it into our reduction to show our main result of Theorem~\ref{thm:main}. 

\section{The Ultra-Wide Word RAM Model}\label{sec:uwram_model}
The \emph{word RAM} model of computation~\cite{Hagerup1998} consists of an unbounded memory of $w$-bit words and a standard instruction set including arithmetic, boolean, and bitwise operations (denoted `$\&$,' `$|$,' and `$\sim$' for \textit{and}, \textit{or} and \textit{not}) and shifts (denoted `$\gg$' and `$\ll$') such as those available in standard programming languages (e.g., C).  We assume that we can store a pointer to the input in a single word and hence $w \geq \log n$, where $n$ is the input size. The time complexity of a word RAM algorithm is the number of instructions, and the space is the number of words the algorithm stores. For our main  result, we focus here on the space used in addition to input array. 

The \emph{ultra-wide word RAM} (UWRAM) model of computation~\cite{FLNS2015} extends the word RAM model with special \emph{ultrawords} of $w^2$ bits. As in~\cite{FLNS2015}, we distinguish between the \emph{restricted UWRAM} that supports a minimal set of instructions on ultrawords consisting of addition, subtraction, shifts, and bitwise boolean operations, and the \emph{multiplication UWRAM} that additionally supports multiplications. We extend the notation for bitwise operations and shifts to ultrawords. The UWRAM (restricted and multiplication) also supports contiguous and scattered memory access operations, as described below. The time complexity is the number of instructions (on standard words or ultrawords), and the space complexity is the number of words used by the algorithms, where each ultraword is counted as $w$ words. The UWRAM model captures (and idealizes) modern vector processing architectures~\cite{CRDI2007, LNOM2008, Reinders2013, Stephensetal2017}. See Farzan et al.~\cite{FLNS2015} for a detailed discussion of the applicability of the UWRAM model. 

\begin{figure}[t]
    \centering
    \includegraphics[scale=0.6]{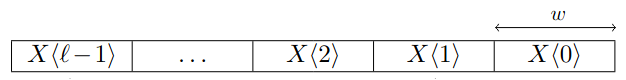}
    \caption{The layout of a word sequence $X$.}
    \label{fig:ultraword_example}
\end{figure}

\subsection{Instructions and Componentwise Operations} 
Recall that ultrawords consist of $w^2$ bits. We often use ultrawords to store and manipulate small sequences of $O(w)$ words. A \emph{word sequence} $X$ of length $\ell$ is a sequence of $\ell$ words (also called the components of $X$). We number the words from right to left starting from $0$ and use the notation $X\angle{i}$ to denote the $i$th word in $X$ (see Figure \ref{fig:ultraword_example}). 

We define common operations on word sequences that we will use later. Let $X$ and $Y$ be word sequences of length $\ell$. The \emph{componentwise addition} of $X$ and $Y$ is the word sequence $Z$ such that $Z\angle{i} = X\angle{i} + Y\angle{i}$. The \emph{componentwise comparison} of $X$ and $Y$ is the word sequence $Z$ such that $Z\angle{i} = 1$ if $X\angle{i} < Y\angle{i}$ and $0$ otherwise. Given another word sequence $I$ of length $\ell$, where each word is either $0$ and $1$ (we will call this a \emph{binary word sequence}), the \emph{componentwise extract} of $X$ wrt. $I$ is the word sequence $Z$ such that $Z\angle{i} = X\angle{i}$ if $I\angle{i} = 1$ and $Z\angle{i} = 0$ otherwise. We can also \emph{concatenate} $X$ and $Y$, denoted $X \cdot Y$ producing the length $2\ell$ word sequence $X\angle{\ell -1} \cdots X\angle{0}Y\angle{\ell -1} \cdots Y\angle{0}$ or \emph{split} $X$ at any point $k$ into word sequences $X\angle{\ell -1} \cdots X\angle{k+1}$ and $X\angle{k} \cdots X\angle{0}$. Note that using two split operations, we can compute any contiguous subsequence $X\angle{i}\cdots X\angle{j}$ of $X$. All these operations can be implemented in constant time for word sequences of length $O(w)$ on the restricted UWRAM using standard-word level parallelism techniques~\cite{Hagerup1998, BGS2022, BGS2024}. Note that we can manipulate word sequences with length $\leq cw$, for some constant $c > 1$,  by storing them in $c$ ultrawords and simulating operations on them in constant time. 

The UWRAM also supports a \emph{compress} operation that takes a binary word sequence $I$ of length $\ell$ and constructs the bitstring of the $\ell$ bits of $I$. The inverse \emph{spread} operation takes a bitstring of length $\ell$ and constructs the corresponding binary word sequence of length $\ell$. This is the UWRAM model that we will use throughout the rest of the paper. Note that these operations are widely supported directly in modern vector processing architectures. see Figure \ref{fig:comp-wise-ops}.

Given a word sequence of $X$ of length $\ell$, we define the \emph{prefix minimum} of $X$, denoted $\pmini(X)$, to be the word sequence $P$ of length $\ell$ such that $P\angle{i} = \mini(X\angle{i}, \ldots, X\angle{0})$.  
We also define the \emph{minimum} of $X$, denoted $\mini(X)$, as the smallest entry among all entries in $X$. Note that we can use a prefix minimum algorithm to compute the minimum. The prefix minimum operation is central in our solutions, and as discussed, we will show how to implement it in $O(\log \log \ell)$ time. 

Some of our operations require precomputed constant word sequences, which we assume are available (e.g., computed at "compile-time"). If not, we can compute those needed for Theorem~\ref{thm:main} in $\log^{O(1)} n$ time, which is negligible. 

\begin{figure}
    \centering
\begin{tikzpicture}
	\begin{pgfonlayer}{nodelayer}
		\node  (0) at (-7, 2.5) {};
		\node  (1) at (-6.5, 2.25) {4};
		\node  (2) at (-5.5, 2.25) {3};
		\node  (3) at (-4.5, 2.25) {9};
		\node  (4) at (-3.5, 2.25) {2};
		\node  (5) at (-7, 2) {};
		\node  (6) at (-3, 2) {};
		\node  (7) at (-3, 2.5) {};
		\node  (8) at (-4, 2.5) {};
		\node  (9) at (-4, 2) {};
		\node  (10) at (-5, 2.5) {};
		\node  (11) at (-5, 2) {};
		\node  (12) at (-6, 2.5) {};
		\node  (13) at (-6, 2) {};
		\node  (14) at (-7, 1.5) {};
		\node  (15) at (-6.5, 1.25) {1};
		\node  (16) at (-5.5, 1.25) {6};
		\node  (17) at (-4.5, 1.25) {4};
		\node  (18) at (-3.5, 1.25) {2};
		\node  (19) at (-7, 1) {};
		\node  (20) at (-3, 1) {};
		\node  (21) at (-3, 1.5) {};
		\node  (22) at (-4, 1.5) {};
		\node  (23) at (-4, 1) {};
		\node  (24) at (-5, 1.5) {};
		\node  (25) at (-5, 1) {};
		\node  (26) at (-6, 1.5) {};
		\node  (27) at (-6, 1) {};
		\node  (28) at (-7, 0.5) {};
		\node  (29) at (-6.5, 0.25) {5};
		\node  (30) at (-5.5, 0.25) {9};
		\node  (31) at (-4.5, 0.25) {13};
		\node  (32) at (-3.5, 0.25) {4};
		\node  (33) at (-7, 0) {};
		\node  (34) at (-3, 0) {};
		\node  (35) at (-3, 0.5) {};
		\node  (36) at (-4, 0.5) {};
		\node  (37) at (-4, 0) {};
		\node  (38) at (-5, 0.5) {};
		\node  (39) at (-5, 0) {};
		\node  (40) at (-6, 0.5) {};
		\node  (41) at (-6, 0) {};
		\node  (42) at (-5, 1.75) {+};
		\node  (43) at (-5, 0.75) {=};
		\node  (44) at (3.25, 2.5) {};
		\node  (45) at (3.75, 2.25) {4};
		\node  (46) at (4.75, 2.25) {3};
		\node  (47) at (5.75, 2.25) {9};
		\node  (48) at (6.75, 2.25) {2};
		\node  (49) at (3.25, 2) {};
		\node  (50) at (7.25, 2) {};
		\node  (51) at (7.25, 2.5) {};
		\node  (52) at (6.25, 2.5) {};
		\node  (53) at (6.25, 2) {};
		\node  (54) at (5.25, 2.5) {};
		\node  (55) at (5.25, 2) {};
		\node  (56) at (4.25, 2.5) {};
		\node  (57) at (4.25, 2) {};
		\node  (58) at (3.25, 1.5) {};
		\node  (59) at (3.75, 1.25) {1};
		\node  (60) at (4.75, 1.25) {0};
		\node  (61) at (5.75, 1.25) {1};
		\node  (62) at (6.75, 1.25) {0};
		\node  (63) at (3.25, 1) {};
		\node  (64) at (7.25, 1) {};
		\node  (65) at (7.25, 1.5) {};
		\node  (66) at (6.25, 1.5) {};
		\node  (67) at (6.25, 1) {};
		\node  (68) at (5.25, 1.5) {};
		\node  (69) at (5.25, 1) {};
		\node  (70) at (4.25, 1.5) {};
		\node  (71) at (4.25, 1) {};
		\node  (72) at (3.25, 0.5) {};
		\node  (73) at (3.75, 0.25) {4};
		\node  (74) at (4.75, 0.25) {0};
		\node  (75) at (5.75, 0.25) {9};
		\node  (76) at (6.75, 0.25) {0};
		\node  (77) at (3.25, 0) {};
		\node  (78) at (7.25, 0) {};
		\node  (79) at (7.25, 0.5) {};
		\node  (80) at (6.25, 0.5) {};
		\node  (81) at (6.25, 0) {};
		\node  (82) at (5.25, 0.5) {};
		\node  (83) at (5.25, 0) {};
		\node  (84) at (4.25, 0.5) {};
		\node  (85) at (4.25, 0) {};
		\node  (86) at (3.75, 1.75) {extract};
		\node  (87) at (5.25, 0.75) {=};
		\node  (88) at (-2, 2.5) {};
		\node  (89) at (-1.5, 2.25) {4};
		\node  (90) at (-0.5, 2.25) {3};
		\node  (91) at (0.5, 2.25) {9};
		\node  (92) at (1.5, 2.25) {2};
		\node  (93) at (-2, 2) {};
		\node  (94) at (2, 2) {};
		\node  (95) at (2, 2.5) {};
		\node  (96) at (1, 2.5) {};
		\node  (97) at (1, 2) {};
		\node  (98) at (0, 2.5) {};
		\node  (99) at (0, 2) {};
		\node  (100) at (-1, 2.5) {};
		\node  (101) at (-1, 2) {};
		\node  (102) at (-2, 1.5) {};
		\node  (103) at (-1.5, 1.25) {1};
		\node  (104) at (-0.5, 1.25) {6};
		\node  (105) at (0.5, 1.25) {4};
		\node  (106) at (1.5, 1.25) {2};
		\node  (107) at (-2, 1) {};
		\node  (108) at (2, 1) {};
		\node  (109) at (2, 1.5) {};
		\node  (110) at (1, 1.5) {};
		\node  (111) at (1, 1) {};
		\node  (112) at (0, 1.5) {};
		\node  (113) at (0, 1) {};
		\node  (114) at (-1, 1.5) {};
		\node  (115) at (-1, 1) {};
		\node  (116) at (-2, 0.5) {};
		\node  (117) at (-1.5, 0.25) {0};
		\node  (118) at (-0.5, 0.25) {1};
		\node  (119) at (0.5, 0.25) {0};
		\node  (120) at (1.5, 0.25) {1};
		\node  (121) at (-2, 0) {};
		\node  (122) at (2, 0) {};
		\node  (123) at (2, 0.5) {};
		\node  (124) at (1, 0.5) {};
		\node  (125) at (1, 0) {};
		\node  (126) at (0, 0.5) {};
		\node  (127) at (0, 0) {};
		\node  (128) at (-1, 0.5) {};
		\node  (129) at (-1, 0) {};
		\node  (130) at (0, 1.75) {$\leq$};
		\node  (131) at (0, 0.75) {=};
		\node  (132) at (-4.75, -1) {};
		\node  (133) at (-4.25, -1.25) {1};
		\node  (134) at (-3.25, -1.25) {0};
		\node  (135) at (-2.25, -1.25) {1};
		\node  (136) at (-1.25, -1.25) {0};
		\node  (137) at (-4.75, -1.5) {};
		\node  (138) at (-0.75, -1.5) {};
		\node  (139) at (-0.75, -1) {};
		\node  (140) at (-1.75, -1) {};
		\node  (141) at (-1.75, -1.5) {};
		\node  (142) at (-2.75, -1) {};
		\node  (143) at (-2.75, -1.5) {};
		\node  (144) at (-3.75, -1) {};
		\node  (145) at (-3.75, -1.5) {};
		\node  (146) at (1.25, -1) {};
		\node  (147) at (2.25, -1) {};
		\node  (148) at (1.25, -1.5) {};
		\node  (149) at (2.25, -1.5) {};
		\node  (150) at (1.75, -1.25) {1010};
		\node  (151) at (-0.25, -1.5) {};
		\node  (152) at (1, -1.5) {};
		\node  (153) at (-0.25, -1) {};
		\node  (154) at (1, -1) {};
		\node  (155) at (0.5, -1.75) {compress};
		\node  (156) at (0.5, -0.75) {spread};
	\end{pgfonlayer}
	\begin{pgfonlayer}{edgelayer}
		\draw (5.center)
			 to (0.center)
			 to (7.center)
			 to (6.center)
			 to cycle;
		\draw (12.center) to (13.center);
		\draw (10.center) to (11.center);
		\draw (8.center) to (9.center);
		\draw (19.center)
			 to (14.center)
			 to (21.center)
			 to (20.center)
			 to cycle;
		\draw (26.center) to (27.center);
		\draw (24.center) to (25.center);
		\draw (22.center) to (23.center);
		\draw (33.center)
			 to (28.center)
			 to [in=180, out=0] (35.center)
			 to (34.center)
			 to cycle;
		\draw (40.center) to (41.center);
		\draw (38.center) to (39.center);
		\draw (36.center) to (37.center);
		\draw (49.center)
			 to (44.center)
			 to (51.center)
			 to (50.center)
			 to cycle;
		\draw (56.center) to (57.center);
		\draw (54.center) to (55.center);
		\draw (52.center) to (53.center);
		\draw (63.center)
			 to (58.center)
			 to (65.center)
			 to (64.center)
			 to cycle;
		\draw (70.center) to (71.center);
		\draw (68.center) to (69.center);
		\draw (66.center) to (67.center);
		\draw (77.center)
			 to (72.center)
			 to [in=180, out=0] (79.center)
			 to (78.center)
			 to cycle;
		\draw (84.center) to (85.center);
		\draw (82.center) to (83.center);
		\draw (80.center) to (81.center);
		\draw (93.center)
			 to (88.center)
			 to (95.center)
			 to (94.center)
			 to cycle;
		\draw (100.center) to (101.center);
		\draw (98.center) to (99.center);
		\draw (96.center) to (97.center);
		\draw (107.center)
			 to (102.center)
			 to (109.center)
			 to (108.center)
			 to cycle;
		\draw (114.center) to (115.center);
		\draw (112.center) to (113.center);
		\draw (110.center) to (111.center);
		\draw (121.center)
			 to (116.center)
			 to [in=180, out=0] (123.center)
			 to (122.center)
			 to cycle;
		\draw (128.center) to (129.center);
		\draw (126.center) to (127.center);
		\draw (124.center) to (125.center);
		\draw (137.center)
			 to (132.center)
			 to (139.center)
			 to (138.center)
			 to cycle;
		\draw (144.center) to (145.center);
		\draw (142.center) to (143.center);
		\draw (140.center) to (141.center);
		\draw (149.center) to (147.center);
		\draw (147.center) to (146.center);
		\draw (146.center) to (148.center);
		\draw (148.center) to (149.center);
		\draw [style=Right] (151.center) to (152.center);
		\draw [style=Right] (154.center) to (153.center);
	\end{pgfonlayer}
\end{tikzpicture}
    \caption{Examples of component-wise instructions on the UWRAM model.}
    \label{fig:comp-wise-ops}
\end{figure}




\subsection{Memory Access}\label{sec:memory_access} The UWRAM supports standard memory access operations that read or write a single word or a sequence of $w$ contiguous words. More interestingly, the UWRAM also supports \emph{scattered} access operations that access $w$ memory locations (not necessarily contiguous) in parallel. Given a word sequence $A$ containing $w$ memory addresses, a \emph{scattered read} loads the contents of the addresses into a word sequence $X$, such that $X\angle{i}$ contains the contents of memory location $A\angle{i}$. Given word sequences $X$ and $A$ of length $O(w)$ a  \emph{scattered write} sets the contents of memory location $A\angle{i}$ to be $X\angle{i}$. We can implement the following \emph{shuffle} operations in constant time using scattered read and write. Given two word sequences $A$ and $X$ of length $\ell$, a \emph{shuffled read} computes the word sequence $Y$, such that $Y\angle{i}=X\angle{A\angle{i}}$. Given word sequences $X$ and $A$ of length $\ell$, a \emph{shuffled write} computes the word sequence $Y$, such that $Y\angle{A\angle{i}}=X\angle{i}$. See Figure \ref{fig:shuffles} for an example of a shuffled read and a shuffled write. Scattered memory accesses capture the memory model used in IBM's \emph{Cell} architecture~\cite{CRDI2007}. They also appear (e.g., \texttt{vpgatherdd}) in Intel's AVX vector extension~\cite{Reinders2013}. Scattered memory access operations were also proposed by Larsen and Pagh~\cite{LP2012} in the context of the I/O model of computation. 
Note that while the addresses for scattered writes must be distinct, we can read simultaneously from the same address. 
\subsection{Extended Ultrawords}\label{sec:extendedultrawords}
Occasionally, we need to work with intermediate integers of precision $cw$ bits for a constant $c>1$ (see the end of Section~ \ref{subsec:prefixminsmallwords}). In this case we use $c$ contiguous words to represent an integer, and adjust the word-level parallelism masks of ~\cite{Hagerup1998, BGS2022, BGS2024} accordingly to implement the operations of the previous subsections (for example, a mask $(10^{w-1})^w$ becomes $(10^{cw-1})^w$, of $cw$ words). To simulate compress we need to bring all the words that contain the bits to compress to contiguous positions in memory. For this we use a shuffled write instead, where the write addresses are given by a word sequence that contains the indices $0, 1, \ldots, w-1$ separated by $c-1$ zero words between each index. The spread operation can be simulated in a similar way. All the algorithms in this paper can be simulated with $c=2$.\footnote{If we wish to use all $w$ bits of precision of an integer, we can either implement the test bit manipulation inside a single word, or use a second word of precision for the test bits, which drives the maximum precision up to $3$ words per integer.}

\begin{figure}
    \centering
\begin{tikzpicture}
	\begin{pgfonlayer}{nodelayer}
		\node  (0) at (-6.5, 2.75) {};
		\node  (1) at (-6, 2.5) {9};
		\node  (4) at (-5, 2.5) {3};
		\node  (5) at (-6.5, 2.25) {};
		\node  (6) at (-0.5, 2.25) {};
		\node  (7) at (-0.5, 2.75) {};
		\node  (8) at (-5.5, 2.75) {};
		\node  (9) at (-5.5, 2.25) {};
		\node  (44) at (-4.5, 2.75) {};
		\node  (45) at (-4.5, 2.25) {};
		\node  (46) at (-3.5, 2.75) {};
		\node  (47) at (-3.5, 2.25) {};
		\node  (48) at (-4, 2.5) {4};
		\node  (49) at (-3, 2.5) {2};
		\node  (50) at (-2.5, 2.75) {};
		\node  (51) at (-2.5, 2.25) {};
		\node  (52) at (-2, 2.5) {8};
		\node  (53) at (-1.5, 2.75) {};
		\node  (54) at (-1.5, 2.25) {};
		\node  (55) at (-1, 2.5) {5};
		\node  (56) at (-6.5, 1.75) {};
		\node  (57) at (-6, 1.5) {4};
		\node  (60) at (-5, 1.5) {4};
		\node  (61) at (-6.5, 1.25) {};
		\node  (62) at (-0.5, 1.25) {};
		\node  (63) at (-0.5, 1.75) {};
		\node  (64) at (-5.5, 1.75) {};
		\node  (65) at (-5.5, 1.25) {};
		\node  (70) at (-4.5, 1.75) {};
		\node  (71) at (-4.5, 1.25) {};
		\node  (72) at (-3.5, 1.75) {};
		\node  (73) at (-3.5, 1.25) {};
		\node  (74) at (-4, 1.5) {2};
		\node  (75) at (-3, 1.5) {2};
		\node  (76) at (-2.5, 1.75) {};
		\node  (77) at (-2.5, 1.25) {};
		\node  (78) at (-2, 1.5) {0};
		\node  (79) at (-1.5, 1.75) {};
		\node  (80) at (-1.5, 1.25) {};
		\node  (81) at (-1, 1.5) {0};
		\node  (82) at (-7, 2.5) {$X$};
		\node  (83) at (-7, 1.5) {$A$};
		\node  (84) at (-7, 0.5) {$Y$};
		\node  (85) at (-6.5, 0.75) {};
		\node  (86) at (-6, 0.5) {3};
		\node  (89) at (-5, 0.5) {3};
		\node  (90) at (-6.5, 0.25) {};
		\node  (91) at (-0.5, 0.25) {};
		\node  (92) at (-0.5, 0.75) {};
		\node  (93) at (-5.5, 0.75) {};
		\node  (94) at (-5.5, 0.25) {};
		\node  (99) at (-4.5, 0.75) {};
		\node  (100) at (-4.5, 0.25) {};
		\node  (101) at (-3.5, 0.75) {};
		\node  (102) at (-3.5, 0.25) {};
		\node  (103) at (-4, 0.5) {2};
		\node  (104) at (-3, 0.5) {2};
		\node  (105) at (-2.5, 0.75) {};
		\node  (106) at (-2.5, 0.25) {};
		\node  (107) at (-2, 0.5) {5};
		\node  (108) at (-1.5, 0.75) {};
		\node  (109) at (-1.5, 0.25) {};
		\node  (110) at (-1, 0.5) {5};
		\node  (111) at (1, 2.75) {};
		\node  (112) at (1.5, 2.5) {9};
		\node  (115) at (2.5, 2.5) {3};
		\node  (116) at (1, 2.25) {};
		\node  (117) at (7, 2.25) {};
		\node  (118) at (7, 2.75) {};
		\node  (119) at (2, 2.75) {};
		\node  (120) at (2, 2.25) {};
		\node  (125) at (3, 2.75) {};
		\node  (126) at (3, 2.25) {};
		\node  (127) at (4, 2.75) {};
		\node  (128) at (4, 2.25) {};
		\node  (129) at (3.5, 2.5) {4};
		\node  (130) at (4.5, 2.5) {2};
		\node  (131) at (5, 2.75) {};
		\node  (132) at (5, 2.25) {};
		\node  (133) at (5.5, 2.5) {8};
		\node  (134) at (6, 2.75) {};
		\node  (135) at (6, 2.25) {};
		\node  (136) at (6.5, 2.5) {5};
		\node  (137) at (1, 1.75) {};
		\node  (138) at (1.5, 1.5) {5};
		\node  (141) at (2.5, 1.5) {1};
		\node  (142) at (1, 1.25) {};
		\node  (143) at (7, 1.25) {};
		\node  (144) at (7, 1.75) {};
		\node  (145) at (2, 1.75) {};
		\node  (146) at (2, 1.25) {};
		\node  (151) at (3, 1.75) {};
		\node  (152) at (3, 1.25) {};
		\node  (153) at (4, 1.75) {};
		\node  (154) at (4, 1.25) {};
		\node  (155) at (3.5, 1.5) {2};
		\node  (156) at (4.5, 1.5) {0};
		\node  (157) at (5, 1.75) {};
		\node  (158) at (5, 1.25) {};
		\node  (159) at (5.5, 1.5) {4};
		\node  (160) at (6, 1.75) {};
		\node  (161) at (6, 1.25) {};
		\node  (162) at (6.5, 1.5) {3};
		\node  (163) at (0.5, 2.5) {$X$};
		\node  (164) at (0.5, 1.5) {$A$};
		\node  (165) at (0.5, 0.5) {$Y$};
		\node  (166) at (1, 0.75) {};
		\node  (167) at (1.5, 0.5) {9};
		\node  (170) at (2.5, 0.5) {8};
		\node  (171) at (1, 0.25) {};
		\node  (172) at (7, 0.25) {};
		\node  (173) at (7, 0.75) {};
		\node  (174) at (2, 0.75) {};
		\node  (175) at (2, 0.25) {};
		\node  (180) at (3, 0.75) {};
		\node  (181) at (3, 0.25) {};
		\node  (182) at (4, 0.75) {};
		\node  (183) at (4, 0.25) {};
		\node  (184) at (3.5, 0.5) {5};
		\node  (185) at (4.5, 0.5) {4};
		\node  (186) at (5, 0.75) {};
		\node  (187) at (5, 0.25) {};
		\node  (188) at (5.5, 0.5) {3};
		\node  (189) at (6, 0.75) {};
		\node  (190) at (6, 0.25) {};
		\node  (191) at (6.5, 0.5) {2};
	\end{pgfonlayer}
	\begin{pgfonlayer}{edgelayer}
		\draw (5.center)
			 to (0.center)
			 to (7.center)
			 to (6.center)
			 to cycle;
		\draw (8.center) to (9.center);
		\draw (44.center) to (45.center);
		\draw (46.center) to (47.center);
		\draw (51.center) to (50.center);
		\draw (54.center) to (53.center);
		\draw (61.center) to (56.center);
		\draw (56.center) to (63.center);
		\draw (63.center) to (62.center);
		\draw (62.center) to (61.center);
		\draw (64.center) to (65.center);
		\draw (70.center) to (71.center);
		\draw (72.center) to (73.center);
		\draw (77.center) to (76.center);
		\draw (80.center) to (79.center);
		\draw (90.center) to (85.center);
		\draw (85.center) to (92.center);
		\draw (92.center) to (91.center);
		\draw (91.center) to (90.center);
		\draw (93.center) to (94.center);
		\draw (99.center) to (100.center);
		\draw (101.center) to (102.center);
		\draw (106.center) to (105.center);
		\draw (109.center) to (108.center);
		\draw (116.center) to (111.center);
		\draw (111.center) to (118.center);
		\draw (118.center) to (117.center);
		\draw (117.center) to (116.center);
		\draw (119.center) to (120.center);
		\draw (125.center) to (126.center);
		\draw (127.center) to (128.center);
		\draw (132.center) to (131.center);
		\draw (135.center) to (134.center);
		\draw (142.center) to (137.center);
		\draw (137.center) to (144.center);
		\draw (144.center) to (143.center);
		\draw (143.center) to (142.center);
		\draw (145.center) to (146.center);
		\draw (151.center) to (152.center);
		\draw (153.center) to (154.center);
		\draw (158.center) to (157.center);
		\draw (161.center) to (160.center);
		\draw (171.center) to (166.center);
		\draw (166.center) to (173.center);
		\draw (173.center) to (172.center);
		\draw (172.center) to (171.center);
		\draw (174.center) to (175.center);
		\draw (180.center) to (181.center);
		\draw (182.center) to (183.center);
		\draw (187.center) to (186.center);
		\draw (190.center) to (189.center);
	\end{pgfonlayer}
\end{tikzpicture}
    \caption{Example of a shuffled read (left) and a shuffled write (right).}
    \label{fig:shuffles}
\end{figure}

\section{Range Minimum Tree}\label{sec:rangeminimumtree}
Let $A$ be an array of $n$ $w$-bit integers and assume for simplicity that $n$ is a power of two. The \emph{range minimum tree} $T$ is the perfectly balanced rooted binary tree over $A$ such that the $i$th leaf corresponds to the $i$th entry in $A$. We associate each node $v$ in $T$ with a \emph{weight}, denoted $\weight(v)$. If $v$ is a leaf, the weight is the value represented by the corresponding entry of $A$, and if $v$ is an internal node, the weight is the minimum of the weights of the descendant leaves. Note that $T$ has height $h = O(\log n)$ and uses $O(n)$ space. See Figure \ref{fig:rmq-query}.

For nodes $i,j\in T$, let $\lca(i,j)$ denote the lowest common ancestor of $i$ and $j$. 
To answer an $\rmq(i,j)$ query, we traverse the path from $i$ to $\lca(i,j)$ and from $j$ to $\lca(i,j)$ and return the minimum weight of  $i$, $j$, and of the nodes hanging off to the right and left of these paths (except for the children of $\lca(i,j)$), respectively (see Figure~\ref{fig:rmq-query}). To perform an $\update(i,\alpha)$, we traverse the path from $i$ to the root. At leaf $i$, we set the weight to be $\alpha$, and at each internal node $v$, we set the weight to be the minimum of the weights of the two children of $v$. If we do not modify the weight of a node at some node in the traversal, we may stop since no weights need to be updated on the remaining path. See Figure~\ref{fig:update-query}. Both operations traverse paths of $O(\log n)$ length and use constant time at each node. 
Hence, both operations use $O(\log n)$ time. 

We introduce the following \emph{node sequences} to implement the range minimum tree on the UWRAM efficiently. Let $i$ be an index in $A$ corresponding to the $i$th leaf in $T$, and let $p$ be the path of nodes from $i$ to the root in $T$. The \emph{path sequence} is the sequence of nodes on the path $p$. Define the \emph{left sequence} to be the sequence of nodes that are hanging off to the left of $p$, i.e., a node $v$ is in the left sequence if it is the left child of a node on $p$ and is not on $p$ itself. Similarly, define the \emph{right sequence} and the \emph{off-path sequence} to be the sequence of nodes to the right of $p$ and the sequence of nodes to the left or right of $p$, respectively.
All node sequences are ordered from $i$ to the root in order of increasing height. See Figures \ref{fig:rmq-query} and \ref{fig:update-query}.

We can use the node sequences to describe the traversed nodes during the $\rmq$ and $\update$ operations on the range minimum tree. Consider an $\rmq(i,j)$ with $u = \lca(i,j)$ of depth $d$. Then $\rmq(i,j)$ is the minimum of the weights of the leaves $i$ and $j$ and of the nodes of depth $>d+1$ on the right sequence of $i$ and the left sequence of $j$. Next, consider an $\update(i, \alpha)$ and let $i^p_0, \ldots, i^p_{h-1}$ and $i^o_0, \ldots, i^o_{h-2}$ be the path and off-path sequence, respectively, for $i$. Let $\weight(\cdot)$ and $\weight'(\cdot)$ denote the weight of nodes in $T$ before and after the $\update$. Recall that only the nodes on the path sequence may change. We have that $\weight'(i^p_0) = \alpha$ and $\weight'(i^p_j) = \min(\weight'(i^p_{j-1}, \weight(i^o_{j-1})))$ for $0 < j < h$. If we unfold the recursion, it follows that 

\begin{equation}\label{eq:updaterecursion}
	\weight'(i^p_j) = \min(\weight(i^o_{j-1}), \ldots,  \weight(i^o_{0}), \alpha) \qquad \text{for $0 \leq j < h$.} 
\end{equation}
In other words, the new weights of the nodes on the path sequence are the prefix minimums of the sequence $\weight(i^o_{h-2}), \ldots, \weight(i^o_{0}), \alpha$. 

\begin{figure}[t]
    \centering
    \scalebox{1}{
    \begin{tikzpicture}
	\begin{pgfonlayer}{nodelayer}
		\node [style=Empty, pattern=north east lines,pattern color=black!30] (0) at (0, 2) {1};
		\node  (1) at (0, 2.5) {$v_1$};
		\node [style=Empty, pattern=north east lines,pattern color=black!30] (2) at (-3, 1.25) {1};
		\node [style=Empty, pattern=north east lines,pattern color=black!30] (3) at (2.75, 1) {2};
		\node [style=Empty, pattern=north east lines,pattern color=black!30] (4) at (-4.75, 0.25) {1};
		\node [style=grey fill] (5) at (-1.75, 0.25) {3};
		\node [style=grey fill] (6) at (1.25, 0.25) {2};
		\node [style=Empty, pattern=north east lines,pattern color=black!30] (7) at (4.25, 0.25) {6};
		\node [style=Empty, pattern=north east lines,pattern color=black!30] (8) at (-5.5, -0.75) {1};
		\node [style=grey fill] (9) at (-4, -0.75) {7};
		\node [style=Empty] (10) at (-2.5, -0.75) {5};
		\node [style=Empty] (11) at (-1, -0.75) {3};
		\node [style=Empty] (12) at (0.5, -0.75) {2};
		\node [style=Empty] (13) at (2, -0.75) {4};
		\node [style=Empty,pattern=north east lines,pattern color=black!30] (14) at (3.5, -0.75) {8};
		\node [style=Empty] (15) at (5, -0.75) {6};
		\node  (16) at (-5.5, -0.25) {$v_8$};
		\node  (17) at (-4, -0.25) {$v_9$};
		\node  (18) at (-2.5, -0.25) {$v_{10}$};
		\node  (19) at (-1, -0.25) {$v_{11}$};
		\node  (20) at (0.5, -0.25) {$v_{12}$};
		\node  (21) at (2, -0.25) {$v_{13}$};
		\node  (22) at (3.5, -0.25) {$v_{14}$};
		\node  (23) at (5, -0.25) {$v_{15}$};
		\node  (24) at (-4.5, 0.75) {$v_4$};
		\node  (25) at (-1.75, 0.75) {$v_5$};
		\node  (26) at (1.25, 0.75) {$v_6$};
		\node  (27) at (4.25, 0.75) {$v_7$};
		\node  (28) at (-2.75, 1.75) {$v_2$};
		\node  (29) at (2.75, 1.5) {$v_3$};
		\node  (30) at (-5, -2.75) {$i$=1};
		\node  (31) at (4, -2.75) {$j$=12};
		\node [align=left] (33) at (0, -3.25) {Right sequence of $v_{17}$: $v_9, v_5, v_3$};
		\node [align=left] (34) at (0, -3.75) {Left sequence of $v_{29}$: $v_{28}, v_{6}, v_2$};
		\node [style=Empty] (35) at (-5.75, -1.75) {5};
		\node [style=Empty, pattern=north east lines,pattern color=black!30] (36) at (-5, -1.75) {1};
		\node [style=Empty] (37) at (-4.25, -1.75) {9};
		\node [style=Empty] (38) at (-3.5, -1.75) {7};
		\node [style=Empty] (39) at (-2.75, -1.75) {5};
		\node [style=Empty] (40) at (-2, -1.75) {6};
		\node [style=Empty] (41) at (-1.25, -1.75) {3};
		\node [style=Empty] (42) at (-0.5, -1.75) {4};
		\node [style=Empty] (43) at (0.25, -1.75) {4};
		\node [style=Empty] (44) at (1, -1.75) {2};
		\node [style=Empty] (45) at (1.75, -1.75) {6};
		\node [style=Empty] (46) at (2.5, -1.75) {4};
		\node [style=grey fill] (47) at (3.25, -1.75) {9};
		\node [style=Empty, pattern=north east lines,pattern color=black!30] (48) at (4, -1.75) {8};
		\node [style=Empty] (49) at (4.75, -1.75) {6};
		\node [style=Empty] (50) at (5.5, -1.75) {6};
		\node  (51) at (-5.75, -2.25) {$v_{16}$};
		\node  (52) at (-5, -2.25) {$v_{17}$};
		\node  (53) at (-4.25, -2.25) {$v_{18}$};
		\node  (54) at (-3.5, -2.25) {$v_{19}$};
		\node  (55) at (-2.75, -2.25) {$v_{20}$};
		\node  (56) at (-2, -2.25) {$v_{21}$};
		\node  (57) at (-1.25, -2.25) {$v_{22}$};
		\node  (58) at (-0.5, -2.25) {$v_{23}$};
		\node  (59) at (0.25, -2.25) {$v_{24}$};
		\node  (60) at (1, -2.25) {$v_{25}$};
		\node  (61) at (1.75, -2.25) {$v_{26}$};
		\node  (62) at (2.5, -2.25) {$v_{27}$};
		\node  (63) at (3.25, -2.25) {$v_{28}$};
		\node  (64) at (4, -2.25) {$v_{29}$};
		\node  (65) at (4.75, -2.25) {$v_{30}$};
		\node  (66) at (5.5, -2.25) {$v_{31}$};
	\end{pgfonlayer}
	\begin{pgfonlayer}{edgelayer}
		\draw [style=Right] (4) to (8);
		\draw [style=Right] (4) to (9);
		\draw [style=Right] (5) to (10);
		\draw [style=Right] (5) to (11);
		\draw [style=Right] (6) to (12);
		\draw [style=Right] (6) to (13);
		\draw [style=Right] (7) to (14);
		\draw [style=Right] (7) to (15);
		\draw [style=Right] (3) to (6);
		\draw [style=Right] (3) to (7);
		\draw [style=Right] (2) to (4);
		\draw [style=Right] (2) to (5);
		\draw [style=Right] (0) to (2);
		\draw [style=Right] (0) to (3);
		\draw [style=Right] (8) to (35);
		\draw [style=Right] (8) to (36);
		\draw [style=Right] (9) to (37);
		\draw [style=Right] (9) to (38);
		\draw [style=Right] (10) to (39);
		\draw [style=Right] (10) to (40);
		\draw [style=Right] (11) to (41);
		\draw [style=Right] (11) to (42);
		\draw [style=Right] (12) to (43);
		\draw [style=Right] (12) to (44);
		\draw [style=Right] (13) to (45);
		\draw [style=Right] (13) to (46);
		\draw [style=Right] (14) to (47);
		\draw [style=Right] (14) to (48);
		\draw [style=Right] (15) to (49);
		\draw [style=Right] (15) to (50);
	\end{pgfonlayer}
\end{tikzpicture}
}
\caption{An example array $A$, with its range minimum tree. For a query $\rmq(1, 12)$, we illustrate with grey circles the right sequence vertices of $i=1$ and the left sequence vertices of $j=12$ of depth greater than $d+1 = 1$. We draw with dashed lines the nodes of the path sequences of $i$ and $j$.}
\label{fig:rmq-query}
\end{figure}

\begin{figure}
    \centering
    \scalebox{1}{
    \begin{tikzpicture}
	\begin{pgfonlayer}{nodelayer}
		\node [style=Empty, pattern=north east lines,pattern color=black!30] (0) at (0, 2) {1};
		\node  (1) at (0, 2.5) {$v_1$};
		\node [style=Empty, pattern=north east lines,pattern color=black!30] (2) at (-3, 1.25) {1};
		\node [style=grey fill] (3) at (2.75, 1) {2};
		\node [style=grey fill] (4) at (-4.75, 0.25) {1};
		\node [style=Empty, pattern=north east lines,pattern color=black!30] (5) at (-1.75, 0.25) {3};
		\node [style=Empty] (6) at (1.25, 0.25) {2};
		\node [style=Empty] (7) at (4.25, 0.25) {6};
		\node [style=Empty] (8) at (-5.5, -0.75) {1};
		\node [style=Empty] (9) at (-4, -0.75) {7};
		\node [style=Empty, pattern=north east lines,pattern color=black!30] (10) at (-2.5, -0.75) {5};
		\node [style=grey fill] (11) at (-1, -0.75) {3};
		\node [style=Empty] (12) at (0.5, -0.75) {2};
		\node [style=Empty] (13) at (2, -0.75) {4};
		\node [style=Empty] (14) at (3.5, -0.75) {8};
		\node [style=Empty] (15) at (5, -0.75) {6};
		\node  (16) at (-5.5, -0.25) {$v_8$};
		\node  (17) at (-4, -0.25) {$v_9$};
		\node  (18) at (-2.5, -0.25) {$v_{10}$};
		\node  (19) at (-1, -0.25) {$v_{11}$};
		\node  (20) at (0.5, -0.25) {$v_{12}$};
		\node  (21) at (2, -0.25) {$v_{13}$};
		\node  (22) at (3.5, -0.25) {$v_{14}$};
		\node  (23) at (5, -0.25) {$v_{15}$};
		\node  (24) at (-4.5, 0.75) {$v_4$};
		\node  (25) at (-1.75, 0.75) {$v_5$};
		\node  (26) at (1.25, 0.75) {$v_6$};
		\node  (27) at (4.25, 0.75) {$v_7$};
		\node  (28) at (-2.75, 1.75) {$v_2$};
		\node  (29) at (2.75, 1.5) {$v_3$};
		\node  (30) at (-2, -2.75) {$i$=5};
		\node [align=left] (33) at (0, -3.25) {Path sequence of $v_{21}$: $v_{21}, v_{10}, v_5, v_2, v_1$};
		\node [align=left] (34) at (0, -3.75) {Off-path sequence of of $v_{21}$: $v_{20}, v_{11}, v_4, v_3$};
		\node [style=Empty] (35) at (-5.75, -1.75) {5};
		\node [style=Empty] (36) at (-5, -1.75) {1};
		\node [style=Empty] (37) at (-4.25, -1.75) {9};
		\node [style=Empty] (38) at (-3.5, -1.75) {7};
		\node [style=grey fill] (39) at (-2.75, -1.75) {5};
		\node [style=Empty, pattern=north east lines,pattern color=black!30] (40) at (-2, -1.75) {6};
		\node [style=Empty] (41) at (-1.25, -1.75) {3};
		\node [style=Empty] (42) at (-0.5, -1.75) {4};
		\node [style=Empty] (43) at (0.25, -1.75) {4};
		\node [style=Empty] (44) at (1, -1.75) {2};
		\node [style=Empty] (45) at (1.75, -1.75) {6};
		\node [style=Empty] (46) at (2.5, -1.75) {4};
		\node [style=Empty] (47) at (3.25, -1.75) {9};
		\node [style=Empty] (48) at (4, -1.75) {8};
		\node [style=Empty] (49) at (4.75, -1.75) {6};
		\node [style=Empty] (50) at (5.5, -1.75) {6};
		\node  (51) at (-5.75, -2.25) {$v_{16}$};
		\node  (52) at (-5, -2.25) {$v_{17}$};
		\node  (53) at (-4.25, -2.25) {$v_{18}$};
		\node  (54) at (-3.5, -2.25) {$v_{19}$};
		\node  (55) at (-2.75, -2.25) {$v_{20}$};
		\node  (56) at (-2, -2.25) {$v_{21}$};
		\node  (57) at (-1.25, -2.25) {$v_{22}$};
		\node  (58) at (-0.5, -2.25) {$v_{23}$};
		\node  (59) at (0.25, -2.25) {$v_{24}$};
		\node  (60) at (1, -2.25) {$v_{25}$};
		\node  (61) at (1.75, -2.25) {$v_{26}$};
		\node  (62) at (2.5, -2.25) {$v_{27}$};
		\node  (63) at (3.25, -2.25) {$v_{28}$};
		\node  (64) at (4, -2.25) {$v_{29}$};
		\node  (65) at (4.75, -2.25) {$v_{30}$};
		\node  (66) at (5.5, -2.25) {$v_{31}$};
	\end{pgfonlayer}
	\begin{pgfonlayer}{edgelayer}
		\draw [style=Right] (4) to (8);
		\draw [style=Right] (4) to (9);
		\draw [style=Right] (5) to (10);
		\draw [style=Right] (5) to (11);
		\draw [style=Right] (6) to (12);
		\draw [style=Right] (6) to (13);
		\draw [style=Right] (7) to (14);
		\draw [style=Right] (7) to (15);
		\draw [style=Right] (3) to (6);
		\draw [style=Right] (3) to (7);
		\draw [style=Right] (2) to (4);
		\draw [style=Right] (2) to (5);
		\draw [style=Right] (0) to (2);
		\draw [style=Right] (0) to (3);
		\draw [style=Right] (8) to (35);
		\draw [style=Right] (8) to (36);
		\draw [style=Right] (9) to (37);
		\draw [style=Right] (9) to (38);
		\draw [style=Right] (10) to (39);
		\draw [style=Right] (10) to (40);
		\draw [style=Right] (11) to (41);
		\draw [style=Right] (11) to (42);
		\draw [style=Right] (12) to (43);
		\draw [style=Right] (12) to (44);
		\draw [style=Right] (13) to (45);
		\draw [style=Right] (13) to (46);
		\draw [style=Right] (14) to (47);
		\draw [style=Right] (14) to (48);
		\draw [style=Right] (15) to (49);
		\draw [style=Right] (15) to (50);
	\end{pgfonlayer}
\end{tikzpicture}
}
    \caption{Illustration of an $\update(5,\alpha)$ query. We draw with grey circles the off-path sequence vertices, and with dashed lines the path sequence ones. Note how the value of $v_{10}$ should be $\min(v_{20}, \alpha)$, the value of $v_5$ should be $\min(v_{10}, v_{11})=\min(v_{11}, v_{20}, \alpha)$, and the value of $v_2$ should be $\min(v_4, v_5)=\min(v_4, v_{11}, v_{20}, \alpha)$, and so on. These values are the prefix minimum of the grey vertices and $\alpha$.}
    \label{fig:update-query}
\end{figure}

\section{From Range Minimum Queries to Prefix Minimum on the UWRAM} 
In this section, we show that any UWRAM data structure that supports prefix minimum computations on word sequences of length $O(\log n)$ implies a UWRAM solution for the range minimum query problem.   

\begin{theorem}\label{thm:reduction} 
Let $A$ be an array of $n$ $w$-bit integers, and let $t(\ell)$ be the time to compute $\pmini$ on a word sequence of length at most $\ell = O(w)$. Then, we can construct a data structure on the UWRAM in $O(n)$ time that supports $\rmq$ and $\update$ in $O(1 + t(\log n))$ time and uses $O(n/\log n)$ words of space in addition to $A$.
\end{theorem}
 In the next section, we will show how to compute $\pmini$ on a word sequence of length $\ell = O(w)$ in $O(\log \log \ell)$ time, implying the main result of Theorem~\ref{thm:main}.

\subsection{Range Minimum Tree on the UWRAM}\label{sec:uwramrangeminimumtree} 
We first show a simple direct implementation of the range minimum tree that achieves the $\rmq$ and $\update$ time bounds of Theorem~\ref{thm:reduction} but with $O(n \log n)$ space and preprocessing time. 

\paragraph{Data Structure} Our data structure consists of the input array $A$, the range minimum tree $T$ over $A$, and a data structure that supports lowest common ancestor queries on $T$. This data structure can be implemented in linear space and preprocessing time to support $\lca$ queries in constant time~\cite{harel1984fast, AGKR2004, BFC2000}. Furthermore, for each index $i$ in $A$, we store the path sequence, the left path sequence, the right path sequence, and the off-path sequence. The sequences are stored as sequences of pointers to the nodes, and together with the left and right sequences, we also store the sequence of depths of the nodes in the sequence.   

The array, the range minimum tree, and the lowest common ancestor data structure use $O(n)$ space. Each of the $O(n)$ sequences uses $O(\log n)$ space. In total, we use $O(n \log n)$ space and preprocessing time. 

\paragraph{Range Minimum Queries} 
To answer a $\rmq(i,j)$ query, compute $u = \lca(i, j)$ and the depth $d$ of $u$. Read the left and right path sequences of $i$ and $j$ into word sequences denoted $I^r$ and $J^l$, and their corresponding depth sequences into word sequences denoted $DI^r$ and $DJ^l$. Then, construct masks $MI^r$ and $MJ^l$ containing $1$s in the positions in $DI^r$ and $DJ^l$ that are greater than $d+1$ and use these to extract the prefixes $\hat{I^r}$ and $\hat{J^l}$ of $I^r$ and $J^l$, respectively, that contain the nodes that have depth greater than $d+1$. Finally, compute the weights $W\hat{I^r}$ and $W\hat{J^l}$ of the nodes in $\hat{I^r}$ and $\hat{J^l}$ using two scattered reads, and return $\mini(W\hat{I^r} \cdot W\hat{J^l})$. The masks $MI^r$ and $MJ^l$ can be computed in the following way: first we load a word sequence $A$ that has a predefined address $a$ in all the components. We write $d+1$ to the memory address $a$, and then use a scattered read with $A$ to get a word sequence $P$ with all the components set to $d+1$ (this is a general operation that we use to copy a specific value to all components of a word sequence). Finally we compare $DI^r$ to $P$ to get $MI^r$, and repeat with $DI^l$ to obtain $MI^l$. Thus, all operations take constant time, except $\mini$, which uses $O(t(\log n))$ time. In total, we use $O(1 + t(\log n))$~time.

\paragraph{Updates}
To implement $\update(i, \alpha)$, we first set $\weight(i) = \alpha$. We read the off-path sequence into a word sequence, denoted $I^o$, and then compute the corresponding sequence of weights $WI^o$ using a scattered read. We then compute $P = \pmini(WI^o \cdot \alpha)$ and perform a scattered write with $I^p$ and $P$. All operations use constant time, except $\pmini$, which uses $O(t(\log n))$ time. In total, we use $O(1 + t(\log n))$. 

\medskip 
In summary, the UWRAM range minimum tree uses $O(n\log n)$ space and preprocessing time and supports $\rmq$ and $\update$ in $O(1 + t(\log n))$ time.

%
%

\subsection{Reducing Space}\label{subsec:reducing_space}
We now improve the space and preprocessing time to $O(n)$ by using a single level of indirection. 

\paragraph{Data Structure}
We partition $A$ into \emph{blocks} $B_0, \ldots, B_{n/\log n-1}$ each of $\log n$ consecutive entries. We store the minimum of each block in a \emph{block array} $B$ of length $n/\log n$ and construct the UWRAM range minimum data structure from Section~\ref{sec:uwramrangeminimumtree} on $B$. This uses $O((n/\log n) \cdot \log (n /\log n)) = O(n)$ space and preprocessing time.

\paragraph{Range Minimum Queries} 
To answer an $\rmq(i,j)$ query there are two cases: 
\begin{description}
	\item[Case 1: $i$ and $j$ are within the same block.] Let $B_k$, where $k = \floor{i/\log n}$ be the block containing $i$ and $j$ and compute the corresponding local indices $i'= i \mod \log n$ and $j' = j \mod \log n$ in $B_k$. We read $B_k$ and return $\mini(B_k\angle{i'},\cdots B_k\angle{j'})$.
		\item[Case 2: $i$ and $j$ are in different blocks.] Let $B_l, B_{l+1}, \ldots, B_r$ be the blocks covering the range from $i$ to $j$ and let $i'$ and $j'$ be the local indices in $B_l$ and $B_r$. We decompose the range into three parts. We compute the minimum in the leftmost block as $l_{\min} = \mini(B_l\angle{i'} \cdots, B_l\angle{\log n -1})$ and the rightmost block  as  $r_{\min} = \mini(B_r\angle{\log n -1}\cdots B_r\angle{j'})$. We then compute the minimum $m_{\min}$ of the middle blocks using the range minimum tree. Finally, we return $\mini(l_{\min}, c_{\min}, r_{\min})$.  
\end{description}

\paragraph{Updates} Consider an $\update(i, \alpha)$ operation. Let $B_k$, where $k =\floor{i/\log n}$ be the block containing $i$, and let $i'= i \mod \log n$ be the local index in $B_k$. We set $B_k\angle{i'} = \alpha$. We then read $B_k$ and compute $b_{\min} = \mini(B_k)$. If $b_{\min}$ differs, we update $T$ with the new value. 

\medskip 
Both operations use constant time except for (prefix) minimum computations and operations on the range minimum tree that take $O(1 + t(\log (n/\log n))) = O(t(\log n))$ time. Hence, the total time is $O(1 + t(\log n))$, with a data structure of linear space and preprocessing time.

\subsection{Further Reducing Space} \label{subsec:heap}
We now show how to further improve the additional space to $O(n/\log n)$ words. The key idea is to use an implicit binary heap layout of the range minimum tree and efficiently implement the queries and updates on the UWRAM in the layout. We first describe the heap layout and how to adapt queries and updates to it in Section~\ref{subsec:heaplayout}. We then show how to efficiently implement the operation on the UWRAM and combine this with indirection as in Section~\ref{subsec:reducing_space} to achieve Theorem~\ref{thm:reduction}.





\subsubsection{Heap Layout and Operations}\label{subsec:heaplayout}
For simplicity in the presentation, we assume that $n$ is a power of two, and if not we pad the input array $A$. We layout the range minimum tree as a binary heap with a single array $H[1..2n-1]$. Each entry in $H$ stores the weight of a node. The original array $A$ is stored at positions $H[n..2n-1]$ and the internal nodes are stored in heap order in positions $H[1..n-1]$ (we use the convention that $H[1]$ stores the root and $H[0]$ is left empty). Throughout this section, we will identify nodes in the range minimum tree with their index in $H$. For any internal node at index $i$, the left and right child are stored at indices $2i$ and $2i+1$, and for any non-root node at index $i$ the parent is stored at $\floor{i/2}$. 

We first give an overview of the high-level idea of how to implement queries and an updates below and then show how to implement them efficiently on the UWRAM in the next subsection.

\paragraph{Minimum Queries} Consider a query $\rmq(i,j)$. We compute nodes sequences in $H$ directly from the binary representation of the leaves at positions $n+i$ and $n+j$. First, we compute the path $p$ sequence for $n+i$, by simply right shifting $n+i$ one bit at the time (corresponding to repeatedly computing parents by integer division by 2). To compute the right sequence for $n+i$ we take all even nodes in $p$ and compute their siblings by adding $1$ (each left child of a node in $p$ is an even node and its right child is at the next index).  Symmetrically, we compute the left sequence for $n+j$ by taking the odd nodes in the path sequence for $n+j$ and subtracting $1$. 

Next, we compute the depth $d$ of $\lca(n+i,n+j)$ in $H$ as the length of the longest common prefix of the binary representation of $n+i$ and $n+j$. We then extract from the right and left sequences all nodes with depth $>d+1$ or equivalently the nodes that are greater than or equal to $2^{d+2}$. Finally, we complete the computation as in Section~\ref{sec:uwramrangeminimumtree}.

\paragraph{Updates} Consider an $\update(i, \alpha)$ operation.  We first compute the path sequence $p$ for $n+i$ as above. We then compute the off-path sequence by flipping the rightmost bit of each node on the off-path sequence except for the root (thus changing nodes to their siblings). Finally, we complete the computation as in Section~\ref{sec:uwramrangeminimumtree}.



\subsubsection{UWRAM Implementation}\label{subsec:UWRAMimplementation}
Now, we show how to implement the above procedures in the UWRAM efficiently. We describe how to compute the node sequences, the depth of the $\lca$ between two nodes, and how to cut the node sequences appropriately.

\paragraph{Node Sequences} Our goal is to compute the path sequence of a node $i$ given by the word sequence $\langle 1,\ldots,\allowbreak\floor{i/4},\allowbreak\floor{i/2},\allowbreak i\rangle$. On a UWRAM with component-wise multiplication this is straightforward to do in constant time, but more involved on the restricted UWRAM. 

In addition to $H$, we store two arrays $H_A[1..2n]$, $H_B[1..2n]$ defined as follows:
\begin{itemize}
    \item For all odd $i\in [1..2n]$, let $k$ be the largest integer such that $2^ki\leq 2n$. Then, $$H_A[i]=H_A[(i\ll 1)]=H_A[(i\ll 2)]=\ldots = H_A[(i\ll k)]=i.$$
    \item For all even $i\in H_B[1..2n]$, let $k$ be the largest integer such that $2^ki+2^{k}-1\leq 2n$. Then, $$H_B[i]=H_B[(i\ll 1)\mid 1]=H_B[(i\ll 2)\mid 11_2]=\ldots  =H_B[(i\ll k)| 1^k_{\hphantom{0}2}]=i.$$
\end{itemize}

These arrays allow us to remove trailing zeros or ones from the binary representation of a number, by reading from $H_A$ or $H_B$ accordingly. For example, the number $12=1100_2$ has two trailing zeros in its binary representation, and then reading from $H_A$ we obtain $H_A[12]=3=11_2$. Now, to right shift any number $i$ by $k$ bits, we can look at its $(k+1)$-th bit: if it is $1$, then we set the bottom k bits to 0 and then read from $H_A$. Otherwise, we set the bottom $k$ bits to $1$ and read from $H_B$. This returns the value $(i\gg k)$. See Figure \ref{fig:heap} for an example with $n=8$.

To obtain the path sequence for $i$, copy $i$ to all components of a word sequence $I=\angle{i,i,\ldots,i}$ (as shown in section \ref{sec:uwramrangeminimumtree}). We need to right shift the $j$th component by $j$ bits. On a high level, we proceed as follows: depending on whether the $j$th bit of the $j$th word is $1$ or $0$, we set all other lower bits to $0$s or $1$s, respectively. We split the components between the ones we filled with $0$s and $1$s, and then we use shuffled reads to get from $H_A$ and $H_B$ the values with the lower bits removed.

More precisely, we load two masks $M_1\angle{j}=2^j$ and $M_2\angle{j}=1^{w-j}0^j$. We detect the terms where $I\angle{j}$ has the $j$th bit set to $1$ by a greater-than ($>$) comparison between $I\:\&\:M_1$ and the zero word sequence; we store that in a word sequence $I_O$. We extract those components from $I$ in a word sequence $O$, and extract the other components to a word sequence $E$. We also extract from $M_2$ the components indicated by $I_O$ to a word sequence $M_2'$, and from $\neg M_2$ the opposite components, to $M_2''$. Then, we make the word sequences $O'=O\:\&\:M_2'$ and $E'=E\mid M_2''$, which set the lower $j$th bits of each non-zero component of $O$ and $E$ to $0$ or $1$, respectively. Finally, we use $O'$ as the indices for a shuffled read on $H_A$, and $E'$ as the indices for a shuffled read on $H_B$. For the zero components of $O'$ and $E'$, we read a dummy $0$ value from $H_A[0]$ and $H_B[0]$. Merging with an OR operation both results, we get the path sequence of $i$. 

From the path sequence of $i$ we can obtain the left, right and off-path sequences easily. For the right path sequence, we extract the components with even values and add $1$ to them. For the left path sequence, we extract the components with odd values and subtract $1$. For the off-path sequence, we flip the last bit of all the components by applying an XOR with $1$. See Figure \ref{fig:path-sequences} for an example of the computation of the four sequences.

\paragraph{Lowest Common Ancestors} 
To compute the depth $d$ of $\lca(i,j)$, we compute the XOR of $i$ and $j$ and then compute the index of most significant bit in the result. This uses constant time~\cite{FW1993}.  

Given the depth $d$, we create a word sequence $\angle{2^{d+2},\ldots, 2^{d+2}}$ by copying $2^{d+2}$ to all components of a word sequence. We can now crop the right and left sequences to nodes of depth $\geq d+2$ by extracting from them the nodes that satisfy a $\geq$-comparison to $\angle{2^{d+2},\ldots, 2^{d+2}}$.



\medskip

All of the above operations except for (prefix) minimum computation and thus take $O(1+t(\log n))$. The space and preprocessing is $O(n)$. To reduce the space we use a single level of indirection exactly as in Section~\ref{subsec:reducing_space}. Thus, we partition $A$ into \emph{blocks} $B_0, \ldots, B_{n/\log n-1}$ each of $\log n$ consecutive entries, store the minimum of each block in a block array $B$. We the construct the data structure on $B$ and implement the operations as in Section~\ref{subsec:reducing_space}. This reduces the additional space to $O(n/\log n)$ while maintaining the same time for queries and updates. In summary, we have shown Theorem~\ref{thm:reduction}.




\begin{figure}
    \centering
\begin{tikzpicture}
	\begin{pgfonlayer}{nodelayer}
		\node  (0) at (-0.5, 1.75) {5};
		\node  (1) at (0, 1.75) {1};
		\node  (2) at (0.5, 1.75) {9};
		\node  (3) at (1, 1.75) {7};
		\node  (4) at (1.5, 1.75) {5};
		\node  (5) at (2, 1.75) {6};
		\node  (6) at (2.5, 1.75) {3};
		\node  (7) at (3, 1.75) {4};
		\node  (32) at (-4, 1.75) {1};
		\node  (33) at (-3.5, 1.75) {1};
		\node  (34) at (-3, 1.75) {3};
		\node  (35) at (-2.5, 1.75) {1};
		\node  (36) at (-2, 1.75) {7};
		\node  (37) at (-1.5, 1.75) {5};
		\node  (38) at (-1, 1.75) {3};
		\node  (43) at (-5.5, 1.75) {$H = $};
		\node  (49) at (3.5, 1.75) {$]$};
		\node  (50) at (-5.5, 1.25) {$H_A = $};
		\node  (51) at (-5.5, 0.75) {$H_B = $};
		\node  (54) at (3.5, 1.25) {$]$};
		\node  (55) at (3.5, 0.75) {$]$};
		\node  (57) at (-4, 1.25) {1};
		\node  (58) at (-3.5, 1.25) {1};
		\node  (59) at (-3, 1.25) {3};
		\node  (60) at (-2.5, 1.25) {1};
		\node  (61) at (-2, 1.25) {5};
		\node  (62) at (-1.5, 1.25) {3};
		\node  (63) at (-1, 1.25) {7};
		\node  (64) at (-0.5, 1.25) {1};
		\node  (65) at (0, 1.25) {9};
		\node  (66) at (0.5, 1.25) {5};
		\node  (67) at (1, 1.25) {11};
		\node  (68) at (1.5, 1.25) {3};
		\node  (69) at (2, 1.25) {13};
		\node  (70) at (2.5, 1.25) {7};
		\node  (71) at (3, 1.25) {15};
		\node  (72) at (-4, 0.75) {0};
		\node  (73) at (-3.5, 0.75) {2};
		\node  (74) at (-3, 0.75) {0};
		\node  (75) at (-2.5, 0.75) {4};
		\node  (76) at (-2, 0.75) {2};
		\node  (77) at (-1.5, 0.75) {6};
		\node  (78) at (-1, 0.75) {0};
		\node  (79) at (-0.5, 0.75) {8};
		\node  (80) at (0, 0.75) {4};
		\node  (81) at (0.5, 0.75) {10};
		\node  (82) at (1, 0.75) {2};
		\node  (83) at (1.5, 0.75) {12};
		\node  (84) at (2, 0.75) {6};
		\node  (85) at (2.5, 0.75) {14};
		\node  (86) at (3, 0.75) {0};
		\node  (87) at (-5, 1.75) {$[$};
		\node  (88) at (-5, 1.25) {$[$};
		\node  (89) at (-5, 0.75) {$[$};
		\node  (90) at (-4.5, 1.75) {0};
		\node  (91) at (-4.5, 1.25) {0};
		\node  (92) at (-4.5, 0.75) {0};
		\node  (93) at (-4.5, 2.25) {\tiny $0$};
		\node  (94) at (-4, 2.25) {\tiny $1$};
		\node  (95) at (-3.5, 2.25) {\tiny $2$};
		\node  (96) at (0.5, 2.25) {\tiny $10$};
		\node  (97) at (1, 2.25) {\tiny $11$};
		\node  (98) at (-3, 2.25) {\tiny $3$};
		\node  (99) at (-2.5, 2.25) {\tiny $4$};
		\node  (100) at (-2, 2.25) {\tiny $5$};
		\node  (101) at (-1.5, 2.25) {\tiny $6$};
		\node  (102) at (-1, 2.25) {\tiny $7$};
		\node  (103) at (-0.5, 2.25) {\tiny $8$};
		\node  (104) at (0, 2.25) {\tiny $9$};
		\node  (105) at (1.5, 2.25) {\tiny $12$};
		\node  (106) at (2, 2.25) {\tiny $13$};
		\node  (107) at (2.5, 2.25) {\tiny $14$};
		\node  (108) at (3, 2.25) {\tiny $15$};
	\end{pgfonlayer}
\end{tikzpicture}
    \caption{Example heap and data structure for an array $A=[5, 1, 9, 7, 5, 6, 3, 4]$. Note that we use left-to-right order here, since $H, H_A$ and $H_B$ are arrays.}
    \label{fig:heap}
\end{figure}

\begin{figure}
    \centering
\begin{tikzpicture}
	\begin{pgfonlayer}{nodelayer}
		\node  (0) at (-7, 1.75) {$I = $};
		\node  (1) at (-6.5, 1.75) {$\langle$};
		\node  (2) at (-5.25, 1.75) {$1010$};
		\node  (3) at (-4.25, 1.75) {$1010$};
		\node  (4) at (-3.25, 1.75) {$1010$};
		\node  (5) at (-2.25, 1.75) {$1010$};
		\node  (6) at (-1.25, 1.75) {$1010$};
		\node  (7) at (-0.75, 1.75) {$\rangle$};
		\node  (8) at (-7.25, 0.75) {$I\&M_1=$};
		\node  (9) at (-6.5, 0.25) {$\langle$};
		\node  (10) at (-5.25, 0.25) {0};
		\node  (11) at (-4.25, 0.25) {1};
		\node  (12) at (-3.25, 0.25) {0};
		\node  (13) at (-2.25, 0.25) {1};
		\node  (14) at (-1.25, 0.25) {0};
		\node  (15) at (-0.75, 0.25) {$\rangle$};
		\node  (16) at (-1.25, 2.25) {\tiny $0$};
		\node  (17) at (-2.25, 2.25) {\tiny $1$};
		\node  (18) at (-3.25, 2.25) {\tiny $2$};
		\node  (19) at (-4.25, 2.25) {\tiny $3$};
		\node  (20) at (-5.25, 2.25) {\tiny $4$};
		\node  (21) at (-6, 1.75) {$\ldots$};
		\node  (22) at (-6, 0.75) {$\ldots$};
		\node  (23) at (-0.75, -0.25) {$\rangle$};
		\node  (24) at (-0.75, -2.25) {$\rangle$};
		\node  (25) at (-1.25, -0.25) {0};
		\node  (26) at (-2.25, -0.25) {1010};
		\node  (27) at (-3.25, -0.25) {0};
		\node  (28) at (-4.25, -0.25) {1010};
		\node  (29) at (-7, 0.25) {$I_O=$};
		\node  (30) at (-6.5, -0.25) {$\langle$};
		\node  (31) at (-1.25, -2.25) {1010};
		\node  (32) at (-2.25, -2.25) {0};
		\node  (33) at (-3.25, -2.25) {1010};
		\node  (34) at (-4.25, -2.25) {0};
		\node  (35) at (-7, -2.25) {$E=$};
		\node  (36) at (-7, -0.25) {$O=$};
		\node  (37) at (-5.25, 0.75) {0};
		\node  (38) at (-4.25, 0.75) {1000};
		\node  (39) at (-3.25, 0.75) {0};
		\node  (40) at (-2.25, 0.75) {$0010$};
		\node  (41) at (-1.25, 0.75) {0};
		\node  (42) at (-5.25, -0.25) {0};
		\node  (43) at (-5.25, -2.25) {0};
		\node  (44) at (-0.75, 0.75) {$\rangle$};
		\node  (45) at (-6.5, 0.75) {$\langle$};
		\node  (46) at (-6.5, -2.25) {$\langle$};
		\node  (47) at (-1.25, -0.75) {0};
		\node  (48) at (-2.25, -0.75) {$1110$};
		\node  (49) at (-3.25, -0.75) {0};
		\node  (50) at (-4.25, -0.75) {$1000$};
		\node  (51) at (-5.25, -0.75) {0};
		\node  (52) at (-6.5, -0.75) {$\langle$};
		\node  (53) at (-7, -0.75) {$M_2'=$};
		\node  (54) at (-7, -2.75) {$M_2''=$};
		\node  (55) at (-6.5, -2.75) {$\langle$};
		\node  (56) at (-0.75, -0.75) {$\rangle$};
		\node  (57) at (-1.25, -2.75) {$0000$};
		\node  (58) at (-2.25, -2.75) {0};
		\node  (59) at (-3.25, -2.75) {$0011$};
		\node  (60) at (-4.25, -2.75) {0};
		\node  (61) at (-5.25, -2.75) {0};
		\node  (62) at (-0.75, -2.75) {$\rangle$};
		\node  (63) at (-0.75, -1.25) {$\rangle$};
		\node  (64) at (-2.25, -1.25) {1010};
		\node  (65) at (-7, -1.25) {$O'=$};
		\node  (66) at (-3.25, -1.25) {0};
		\node  (67) at (-4.25, -1.25) {1000};
		\node  (68) at (-1.25, -1.25) {0};
		\node  (69) at (-5.25, -1.25) {0};
		\node  (70) at (-0.75, -3.25) {$\rangle$};
		\node  (71) at (-7, -3.25) {$E'=$};
		\node  (72) at (-1.25, -3.25) {1010};
		\node  (73) at (-2.25, -3.25) {0};
		\node  (74) at (-3.25, -3.25) {1011};
		\node  (75) at (-4.25, -3.25) {0};
		\node  (76) at (-5.25, -3.25) {0};
		\node  (77) at (-0.75, -1.75) {$\rangle$};
		\node  (78) at (-7.25, -1.75) {$H_A[O']=$};
		\node  (79) at (-1.25, -1.75) {0};
		\node  (80) at (-2.25, -1.75) {$0101$};
		\node  (81) at (-3.25, -1.75) {0};
		\node  (82) at (-4.25, -1.75) {$0001$};
		\node  (83) at (-5.25, -1.75) {0};
		\node  (84) at (-0.75, -3.75) {$\rangle$};
		\node  (85) at (-1.25, -3.75) {$1010$};
		\node  (86) at (-2.25, -3.75) {0};
		\node  (87) at (-3.25, -3.75) {$0010$};
		\node  (88) at (-4.25, -3.75) {0};
		\node  (89) at (-7.25, -3.75) {$H_B[E']=$};
		\node  (90) at (-5.25, -3.75) {$0$};
		\node  (91) at (-7, -4.25) {$P=$};
		\node  (92) at (-5.25, -4.25) {$0$};
		\node  (93) at (-4.25, -4.25) {$0001$};
		\node  (94) at (-3.25, -4.25) {$0010$};
		\node  (95) at (-2.25, -4.25) {$0101$};
		\node  (96) at (-1.25, -4.25) {1010};
		\node  (97) at (-0.75, -4.25) {$\rangle$};
		\node  (98) at (-6.5, -1.25) {$\langle$};
		\node  (99) at (-6.5, -3.25) {$\langle$};
		\node  (100) at (-6.5, -1.75) {$\langle$};
		\node  (101) at (-6.5, -3.75) {$\langle$};
		\node  (102) at (-6.5, -4.25) {$\langle$};
		\node  (103) at (-7, 1.25) {$M_1=$};
		\node  (104) at (-6, 1.25) {$\ldots$};
		\node  (105) at (-5.25, 1.25) {$10000$};
		\node  (106) at (-4.25, 1.25) {1000};
		\node  (107) at (-3.25, 1.25) {$0100$};
		\node  (108) at (-2.25, 1.25) {$0010$};
		\node  (109) at (-1.25, 1.25) {0000};
		\node  (110) at (-0.75, 1.25) {$\rangle$};
		\node  (111) at (-6.5, 1.25) {$\langle$};
		\node  (112) at (0.5, 1.75) {$P=$};
		\node  (114) at (2, 1.75) {$0001$};
		\node  (115) at (3, 1.75) {$0010$};
		\node  (116) at (4, 1.75) {$0101$};
		\node  (117) at (5, 1.75) {1010};
		\node  (118) at (5.5, 1.75) {$\rangle$};
		\node  (119) at (1, 1.75) {$\langle$};
		\node  (120) at (5, 2.25) {\tiny $0$};
		\node  (121) at (4, 2.25) {\tiny $1$};
		\node  (122) at (3, 2.25) {\tiny $2$};
		\node  (123) at (2, 2.25) {\tiny $3$};
		\node  (124) at (2, 1.25) {0};
		\node  (125) at (3, 1.25) {$0010$};
		\node  (126) at (4, 1.25) {0};
		\node  (127) at (5, 1.25) {1010};
		\node  (128) at (5.5, 1.25) {$\rangle$};
		\node  (129) at (1, 1.25) {$\langle$};
		\node  (130) at (2, 0.75) {0};
		\node  (131) at (3, 0.75) {$0011$};
		\node  (132) at (4, 0.75) {0};
		\node  (133) at (5, 0.75) {1011};
		\node  (134) at (5.5, 0.75) {$\rangle$};
		\node  (135) at (1, 0.75) {$\langle$};
		\node  (136) at (2, 0) {$0001$};
		\node  (137) at (3, 0) {0};
		\node  (138) at (4, 0) {$0101$};
		\node  (139) at (5, 0) {0};
		\node  (140) at (5.5, 0) {$\rangle$};
		\node  (141) at (1, 0) {$\langle$};
		\node  (142) at (0.25, 1.25) {even =};
		\node  (143) at (0.25, 0.75) {right q. =};
		\node  (144) at (0.25, 0) {odd =};
		\node  (145) at (0.25, -0.5) {left q. =};
		\node  (146) at (0.25, -1.25) {off-p. =};
		\node  (147) at (2, -0.5) {0};
		\node  (148) at (3, -0.5) {0};
		\node  (149) at (4, -0.5) {0100};
		\node  (150) at (5, -0.5) {0};
		\node  (151) at (5.5, -0.5) {$\rangle$};
		\node  (152) at (2, -1.25) {0000};
		\node  (153) at (3, -1.25) {0011};
		\node  (154) at (4, -1.25) {0100};
		\node  (155) at (5, -1.25) {1011};
		\node  (156) at (5.5, -1.25) {$\rangle$};
		\node  (157) at (1, -0.5) {$\langle$};
		\node  (158) at (1, -1.25) {$\langle$};
	\end{pgfonlayer}
\end{tikzpicture}
    \caption{Example of a path sequence computation for $i=2$, for an array of size $8$ ($P$ is the output). On the right we show the steps to get the right, left and off-path sequences.}
    \label{fig:path-sequences}
\end{figure}






\section{Computing Prefix Minimum on Word Sequences}\label{sec:prefixminwordsequence}
We now show how to efficiently compute the prefix minimum on word sequences of length $\ell = O(w)$ in $O(\log \log \ell)$ time. We first show how to do so in constant time for word sequences of length $O(\sqrt{w})$. We then show how to implement this algorithm in parallel and then recursively leading to the result. 

Our algorithm often partitions a word sequence into multiple equal-length sequences to work on them in parallel. We define a $b$-way word sequence to be a word sequence $X = X_{s-1} \cdots X_0$ where each subsequence $X_i$ is a \emph{block} of length $b$. Thus, the total length of $X$ is $sb$. We use $X\angle{i,j}$ to denote entry $j$ in block $i$, that is, $X\angle{i,j} = X_i\angle{j}$. 

We will need a subroutine on $b$-way word sequences that we call \emph{left clear}, denoted $\mathsf{lclear}$. We first define it for words. Let $x$ be a bit string, and define $\mathsf{rmz}(x)$ to be the position of the rightmost $0$ bit of $x$. The operation $\mathsf{lclear}(x)$ is defined as $0^{|x|-\mathsf{rmz}(x)}\cdot 1^{\mathsf{rmz}(x)}$. Thus, $\mathsf{lclear}$ ``smears'' the rightmost $0$ of $x$ to the left. To extend the definition to $b$-way word sequences, let $X=X_{s-1}\cdots X_0$ be a binary $b$-way word sequence with $s$ blocks, and length $sb$. Then, $\mathsf{lclear}(X, b)$ is the $b$-way word sequence $Y=Y_{s-1}\cdots Y_0$ such that $Y_i=\mathsf{lclear}(X_i)$.

We can compute $Y=\mathsf{lclear}(X, b)$ in $O(1)$ time as follows. First, observe that for a word $x$ we can compute $\mathsf{lclear}(x)$ with $x\&\overline{x+1}$, see Knuth \cite{KnuthVol4}. To extend this to a $b$-way word sequence $X$, first, we use $\compress$ on $X$ to obtain a bitstring $x$ of length $sb = O(w)$ bits. Then, we compute $y=x\&\overline{x+x'}$, where $x'$ is the pre-computed constant $(0^{b-1}1)^{s}$, and the addition is done component-wise to avoid carries between components (the components are of $b$ bits each). Finally, we use $\spread$ on $y$ to get the desired word sequence $Y$.

\subsection{Prefix Minimum on Small Word Sequences} \label{subsec:prefixminsmallwords}
\begin{figure}[t]
\centering	
\begin{tikzpicture}
	\begin{pgfonlayer}{nodelayer}
		\node  (0) at (-4, 1.5) {$9$};
		\node  (1) at (-3.5, 1.5) {$2$};
		\node  (2) at (-3, 1.5) {$5$};
		\node  (3) at (-2.5, 1.5) {$3$};
		\node  (4) at (-1.75, 1.5) {$9$};
		\node  (5) at (-1.25, 1.5) {$2$};
		\node  (6) at (-0.75, 1.5) {$5$};
		\node  (7) at (-0.25, 1.5) {$3$};
		\node  (8) at (0.5, 1.5) {$9$};
		\node  (9) at (1, 1.5) {$2$};
		\node  (10) at (1.5, 1.5) {$5$};
		\node  (11) at (2, 1.5) {$3$};
		\node  (12) at (2.75, 1.5) {$9$};
		\node  (13) at (3.25, 1.5) {$2$};
		\node  (14) at (3.75, 1.5) {$5$};
		\node  (15) at (4.25, 1.5) {$3$};
		\node  (16) at (2.75, 3.5) {$9$};
		\node  (17) at (3.25, 3.5) {$2$};
		\node  (18) at (3.75, 3.5) {$5$};
		\node  (19) at (4.25, 3.5) {$3$};
		\node  (20) at (-4, 2) {$9$};
		\node  (21) at (-3.5, 2) {$9$};
		\node  (22) at (-3, 2) {$9$};
		\node  (23) at (-2.5, 2) {$9$};
		\node  (24) at (-1.75, 2) {$2$};
		\node  (25) at (-1.25, 2) {$2$};
		\node  (26) at (-0.75, 2) {$2$};
		\node  (27) at (-0.25, 2) {$2$};
		\node  (28) at (0.5, 2) {$5$};
		\node  (29) at (1, 2) {$5$};
		\node  (30) at (1.5, 2) {$5$};
		\node  (31) at (2, 2) {$5$};
		\node  (32) at (2.75, 2) {$3$};
		\node  (33) at (3.25, 2) {$3$};
		\node  (34) at (3.75, 2) {$3$};
		\node  (35) at (4.25, 2) {$3$};
		\node  (36) at (-4, 1) {$1$};
		\node  (37) at (-1.75, 1) {$1$};
		\node  (38) at (-1.25, 1) {$1$};
		\node  (39) at (-0.75, 1) {$1$};
		\node  (40) at (-0.25, 1) {$1$};
		\node  (41) at (0.5, 1) {$1$};
		\node  (42) at (1, 1) {$0$};
		\node  (43) at (1.5, 1) {$1$};
		\node  (44) at (2, 1) {$0$};
		\node  (45) at (-3.5, 1) {$0$};
		\node  (46) at (-3, 1) {$0$};
		\node  (47) at (-2.5, 1) {$0$};
		\node  (48) at (2.75, 1) {$1$};
		\node  (49) at (3.25, 1) {$0$};
		\node  (50) at (3.75, 1) {$1$};
		\node  (51) at (4.25, 1) {$1$};
		\node  (52) at (-4, 0.5) {$0$};
		\node  (53) at (-3.5, 0.5) {$0$};
		\node  (54) at (-3, 0.5) {$0$};
		\node  (55) at (-2.5, 0.5) {$0$};
		\node  (56) at (0.5, 0.5) {$0$};
		\node  (57) at (1, 0.5) {$0$};
		\node  (58) at (1.5, 0.5) {$0$};
		\node  (59) at (2, 0.5) {$0$};
		\node  (60) at (-1.75, 0.5) {$1$};
		\node  (61) at (-1.25, 0.5) {$1$};
		\node  (62) at (-0.75, 0.5) {$1$};
		\node  (63) at (-0.25, 0.5) {$1$};
		\node  (64) at (4.25, 0.5) {$1$};
		\node  (65) at (3.75, 0.5) {$1$};
		\node  (66) at (3.25, 0.5) {$0$};
		\node  (67) at (2.75, 0.5) {$0$};
		\node  (68) at (-4, -0.5) {$0$};
		\node  (69) at (-3.5, -0.5) {$0$};
		\node  (70) at (-3, -0.5) {$0$};
		\node  (71) at (-2.5, -0.5) {$0$};
		\node  (72) at (0.5, -0.5) {$0$};
		\node  (73) at (1, -0.5) {$0$};
		\node  (74) at (1.5, -0.5) {$0$};
		\node  (75) at (2, -0.5) {$0$};
		\node  (76) at (-1.75, -0.5) {$1$};
		\node  (77) at (-1.25, -0.5) {$1$};
		\node  (78) at (-0.75, -0.5) {$0$};
		\node  (79) at (-0.25, -0.5) {$0$};
		\node  (80) at (4.25, -0.5) {$1$};
		\node  (81) at (3.75, -0.5) {$1$};
		\node  (82) at (3.25, -0.5) {$0$};
		\node  (83) at (2.75, -0.5) {$0$};
		\node  (84) at (-4, -1.5) {$0$};
		\node  (85) at (-3.5, -1.5) {$0$};
		\node  (86) at (-3, -1.5) {$0$};
		\node  (87) at (-2.5, -1.5) {$0$};
		\node  (88) at (0.5, -1.5) {$0$};
		\node  (89) at (1, -1.5) {$0$};
		\node  (90) at (1.5, -1.5) {$0$};
		\node  (91) at (2, -1.5) {$0$};
		\node  (92) at (-1.75, -1.5) {$3$};
		\node  (93) at (-1.25, -1.5) {$2$};
		\node  (94) at (-0.75, -1.5) {$0$};
		\node  (95) at (-0.25, -1.5) {$0$};
		\node  (96) at (4.25, -1.5) {$0$};
		\node  (97) at (3.75, -1.5) {$1$};
		\node  (98) at (3.25, -1.5) {$0$};
		\node  (99) at (2.75, -1.5) {$0$};
		\node  (100) at (-4, -2.5) {$19$};
		\node  (101) at (-3.5, -2.5) {$18$};
		\node  (102) at (-3, -2.5) {$17$};
		\node  (103) at (-2.5, -2.5) {$16$};
		\node  (104) at (0.5, -2.5) {$11$};
		\node  (105) at (1, -2.5) {$10$};
		\node  (106) at (1.5, -2.5) {$9$};
		\node  (107) at (2, -2.5) {$8$};
		\node  (108) at (-1.75, -2.5) {$0$};
		\node  (109) at (-1.25, -2.5) {$0$};
		\node  (110) at (-0.75, -2.5) {$13$};
		\node  (111) at (-0.25, -2.5) {$12$};
		\node  (112) at (4.25, -2.5) {$0$};
		\node  (113) at (3.75, -2.5) {$0$};
		\node  (114) at (3.25, -2.5) {$6$};
		\node  (115) at (2.75, -2.5) {$7$};
		\node  (128) at (4.25, -1) {$0$};
		\node  (129) at (3.75, -1) {$1$};
		\node  (130) at (3.25, -1) {$2$};
		\node  (131) at (2.75, -1) {$3$};
		\node  (132) at (2.75, -3.5) {$2$};
		\node  (133) at (3.25, -3.5) {$2$};
		\node  (134) at (3.75, -3.5) {$3$};
		\node  (135) at (4.25, -3.5) {$3$};
		\node  (136) at (4.75, 3.5) {$\rangle$};
		\node  (137) at (4.75, 2) {$\rangle$};
		\node  (138) at (4.75, 1.5) {$\rangle$};
		\node  (139) at (4.75, 1) {$\rangle$};
		\node  (140) at (4.75, 0.5) {$\rangle$};
		\node  (141) at (4.75, -0.5) {$\rangle$};
		\node  (142) at (4.75, -1.5) {$\rangle$};
		\node  (143) at (4.75, -2.5) {$\rangle$};
		\node  (144) at (4.75, -1) {$\rangle$};
		\node  (145) at (4.75, -3.5) {$\rangle$};
		\node [align=right] (146) at (-4.75, 3.5) {$X=\hphantom{\langle}$};
		\node [align=right] (147) at (-4.75, 2) {$\widetilde{X}=\langle$};
		\node [align=right] (148) at (-4.75, 1.5) {$\widehat{X}=\langle$};
		\node [align=right] (149) at (-4.75, 1) {$C=\langle$};
		\node [align=right] (150) at (-4.75, 0.5) {$D=\langle$};
		\node [align=right] (151) at (-4.75, 0) {$M=\langle$};
		\node [align=right] (152) at (-4.75, -0.5) {$E=\langle$};
		\node [align=right] (153) at (-4.75, -1.5) {$E'=\langle$};
		\node [align=right] (154) at (-4.75, -2.5) {$E''=\langle$};
		\node [align=right] (155) at (-4.75, -1) {$P'=\langle$};
		\node [align=right] (156) at (-4.75, -3.5) {$Y=\hphantom{\langle}$};
		\node  (158) at (-4, 0) {$1$};
		\node  (159) at (-3.5, 0) {$0$};
		\node  (160) at (-3, 0) {$0$};
		\node  (161) at (-2.5, 0) {$0$};
		\node  (162) at (0.5, 0) {$1$};
		\node  (163) at (1, 0) {$1$};
		\node  (164) at (1.5, 0) {$1$};
		\node  (165) at (2, 0) {$0$};
		\node  (166) at (-1.75, 0) {$1$};
		\node  (167) at (-1.25, 0) {$1$};
		\node  (168) at (-0.75, 0) {$0$};
		\node  (169) at (-0.25, 0) {$0$};
		\node  (170) at (4.25, 0) {$1$};
		\node  (171) at (3.75, 0) {$1$};
		\node  (172) at (3.25, 0) {$1$};
		\node  (173) at (2.75, 0) {$1$};
		\node  (174) at (4.75, 0) {$\rangle$};
		\node  (175) at (2, 3.5) {$\langle$};
		\node  (176) at (2, -3.5) {$\langle$};
		\node  (177) at (-4, -3) {$19$};
		\node  (178) at (-3.5, -3) {$18$};
		\node  (179) at (-3, -3) {$17$};
		\node  (180) at (-2.5, -3) {$16$};
		\node  (181) at (0.5, -3) {$11$};
		\node  (182) at (1, -3) {$10$};
		\node  (183) at (1.5, -3) {$9$};
		\node  (184) at (2, -3) {$8$};
		\node  (185) at (-1.75, -3) {$3$};
		\node  (186) at (-1.25, -3) {$2$};
		\node  (187) at (-0.75, -3) {$13$};
		\node  (188) at (-0.25, -3) {$12$};
		\node  (189) at (4.25, -3) {$0$};
		\node  (190) at (3.75, -3) {$1$};
		\node  (191) at (3.25, -3) {$6$};
		\node  (192) at (2.75, -3) {$7$};
		\node  (193) at (4.75, -3) {$\rangle$};
		\node [align=right] (194) at (-4.75, -3) {$P=\langle$};
		\node  (195) at (-4, -2) {$19$};
		\node  (196) at (-3.5, -2) {$18$};
		\node  (197) at (-3, -2) {$17$};
		\node  (198) at (-2.5, -2) {$16$};
		\node  (199) at (0.5, -2) {$11$};
		\node  (200) at (1, -2) {$10$};
		\node  (201) at (1.5, -2) {$9$};
		\node  (202) at (2, -2) {$8$};
		\node  (203) at (-1.75, -2) {$15$};
		\node  (204) at (-1.25, -2) {$14$};
		\node  (205) at (-0.75, -2) {$13$};
		\node  (206) at (-0.25, -2) {$12$};
		\node  (207) at (4.25, -2) {$4$};
		\node  (208) at (3.75, -2) {$5$};
		\node  (209) at (3.25, -2) {$6$};
		\node  (210) at (2.75, -2) {$7$};
		\node  (211) at (4.75, -2) {$\rangle$};
		\node [align=right] (212) at (-4.75, -2) {$P''=\langle$};
		\node  (213) at (2, -1) {$0$};
		\node  (214) at (1.5, -1) {$1$};
		\node  (215) at (1, -1) {$2$};
		\node  (216) at (0.5, -1) {$3$};
		\node  (217) at (-0.25, -1) {$0$};
		\node  (218) at (-0.75, -1) {$1$};
		\node  (219) at (-1.25, -1) {$2$};
		\node  (220) at (-1.75, -1) {$3$};
		\node  (221) at (-2.5, -1) {$0$};
		\node  (222) at (-3, -1) {$1$};
		\node  (223) at (-3.5, -1) {$2$};
		\node  (224) at (-4, -1) {$3$};
		\node  (225) at (4.25, 2.5) {$0$};
		\node  (226) at (3.75, 2.5) {$1$};
		\node  (227) at (3.25, 2.5) {$2$};
		\node  (228) at (2.75, 2.5) {$3$};
		\node  (229) at (4.75, 2.5) {$\rangle$};
		\node [align=right] (230) at (-4.75, 2.5) {$\widehat{A}=\langle$};
		\node  (231) at (2, 2.5) {$0$};
		\node  (232) at (1.5, 2.5) {$1$};
		\node  (233) at (1, 2.5) {$2$};
		\node  (234) at (0.5, 2.5) {$3$};
		\node  (235) at (-0.25, 2.5) {$0$};
		\node  (236) at (-0.75, 2.5) {$1$};
		\node  (237) at (-1.25, 2.5) {$2$};
		\node  (238) at (-1.75, 2.5) {$3$};
		\node  (239) at (-2.5, 2.5) {$0$};
		\node  (240) at (-3, 2.5) {$1$};
		\node  (241) at (-3.5, 2.5) {$2$};
		\node  (242) at (-4, 2.5) {$3$};
		\node  (243) at (4.25, 3) {$0$};
		\node  (244) at (3.75, 3) {$0$};
		\node  (245) at (3.25, 3) {$0$};
		\node  (246) at (2.75, 3) {$0$};
		\node  (247) at (4.75, 3) {$\rangle$};
		\node [align=right] (248) at (-4.75, 3) {$\widetilde{A}=\langle$};
		\node  (249) at (2, 3) {$1$};
		\node  (250) at (1.5, 3) {$1$};
		\node  (251) at (1, 3) {$1$};
		\node  (252) at (0.5, 3) {$1$};
		\node  (253) at (-0.25, 3) {$2$};
		\node  (254) at (-0.75, 3) {$2$};
		\node  (255) at (-1.25, 3) {$2$};
		\node  (256) at (-1.75, 3) {$2$};
		\node  (257) at (-2.5, 3) {$3$};
		\node  (258) at (-3, 3) {$3$};
		\node  (259) at (-3.5, 3) {$3$};
		\node  (260) at (-4, 3) {$3$};
	\end{pgfonlayer}
\end{tikzpicture}
  \caption{Our prefix minimum algorithm, for $w=16$ and a word sequence of length $4$.}
  \label{fig:prefixmin}
\end{figure}
We now show how to compute the prefix minimum on a word sequence $X$ of length $\ell = O(\sqrt{w})$. For simplicity, we first assume all entries in $X$ are distinct. 
On a high-level, out algorithm first computes an all-to-all, parallel comparison of the elements of $X$. There are $l^2=O(w)$ comparisons, so we can do the computation in constant time. Then, we use the comparison results to obtain, for each word $X\angle{i}$, what indices of the prefix minimum are equal to $X\angle{i}$. The key observation is that $X\angle{i}$ is the prefix minimum at index $j$ if and only if $X\angle{i}$ is less than or equal to $X\angle{0}, X\angle{1}, \ldots, X\angle{j}$, and $i\leq j$. We can compute both conditions with the results of the all-to-all comparison as we show below. Finally, we copy each word $X\angle{i}$ to the indices of the output where it is the prefix minimum.

Our algorithm proceeds as follows. See also the example in Figure~\ref{fig:prefixmin}. 

\paragraph{Step 1: Compare all pairs of words in $X$} We construct a $b$-way word sequence $C$ containing the results of all pairwise comparisons of words in $X$. To do so, we first construct the word sequences: 
\begin{align*}
	\widetilde{X} &= \underbrace{X\angle{\ell-1} \cdots X\angle{\ell-1}}_\ell \cdot \underbrace{X\angle{\ell-2} \cdots X\angle{\ell-2}}_\ell \cdots \underbrace{X\angle{0} \cdots X\angle{0}}_\ell \\
	\widehat{X}&= \underbrace{X \cdot X \cdots X}_{\ell}	
\end{align*}
We compute these using shuffled read operations on $X$ with the constant word sequences
\begin{align*}
	\widetilde{A} &= \underbrace{\angle{\ell - 1, \ldots \ell-1}}_{\ell}\cdot \underbrace{\angle{\ell - 2, \ldots \ell-2}}_{\ell} \cdots \underbrace{\angle{0, \ldots 0}}_{\ell} \\
	\widehat{A}  &= \underbrace{\angle{\ell-1, \cdots 0} \cdots \angle{\ell-1, \cdots 0}}_{\ell} .
\end{align*} Note that both $\widehat{X}$ and $\widetilde{X}$ have length $\ell^2 = O(w)$. We then do a componentwise comparison of $\widehat{X}$ and $\widetilde{X}$. This produces a word sequence $C$ of length $\ell^2$, which viewed as an $\ell$-way word sequence is defined by: 
\begin{equation}\label{eq:compare} 
C\angle{i,j} = 
\begin{cases}
1 & \text{if $X\angle{i}  \leq X\angle{j}$} \\
0 & \text{otherwise} 
\end{cases}
\end{equation}
Thus, the $i$th block in $C$ stores the comparison of $X\angle{i}$ with all other words in~$X$.

\paragraph{Step 2: Compute Prefix Minima} 
We construct an $\ell$-way word sequence $E$ that contains the positions of the prefix minima. First, compute the $\ell$-way word sequence $D$ such that
\begin{equation}
	D\angle{i,j} = 
		\begin{cases}
			1 & \text{if $C\angle{i,j} = C\angle{i,j-1} = \cdots C\angle{i, 0} = 1$} \\
			0 & \text{otherwise} 
	\end{cases}
\end{equation}
Thus, the block $D_i$ in $D$ takes the rightmost $0$ in $C_i$ and "smears" it to the left. To compute this, we use the constant-time \textsf{lclear} operation. We then mask out all entries in $D$ where $i > j$ using an $\ell$-way mask $M = M_{\ell-1} \cdots M_0$, where $M\angle{i,j} = 1$ if $i \leq j$ and $0$ otherwise. The result is the desired word sequence $E$. By \eqref{eq:compare} we have that $D\angle{i,j} = 1$ iff $X\angle{i} \leq X\angle{j}, X\angle{j-1}, \ldots, X\angle{0}$. Hence, $E\angle{i,j} = 1$ iff $X\angle{i} \leq X\angle{j}, X\angle{j-1}, \ldots, X\angle{0}$ and $i \leq j$, i.e., $X\angle{i}$ is the prefix minimum of $X\angle{j-1}, \ldots, X\angle{0}$.

\paragraph{Step 3: Extract Prefix Minimum} 
We now extract the prefix minima entries indicated by $E$ and then compact them into a word sequence of length $\ell$. The key observation is that $E\angle{i,j} = 1$ implies that $X\angle{i}$ should appear in position $j$ in the final prefix minimum. We use the following constant word sequences:  
\begin{align*}
	P'  &= \underbrace{\angle{\ell-1, \cdots 0} \cdots \angle{\ell-1, \cdots 0}}_{\ell} \\
	P'' &= \angle{\ell^2 + \ell - 1, \ell^2 + \ell - 2, \ldots \ell}
\end{align*}
$P'$ contains in each block $P'\angle{i}$ the indices of the prefix minimum where $X\angle{i}$ has to appear, and $P''$ contains dummy indices (greater than $\ell$) in all other positions. Note how $P'=\widehat{A}$. 
We do a component-wise extraction of the positions in $P'$ wrt.\ $E$ and of $P''$ wrt.\ $\overline{E}$, where $\overline{E}$ is the negation of $E$. We \textit{or} the resulting sequences $E'$ and $E''$ together to get a word sequence $P$. We then do a shuffled write on $\widetilde{X}$ with $P$ and clear out all but the $\ell$ rightmost entries. The resulting word sequence $Y$ of length $\ell$ contains a $X\angle{i}$ in position $j$ iff $E\angle{i,j} = 1$ as desired. Since we assumed all input words are distinct and since $P''$ only contains distinct numbers greater than $\ell-1$, there is no write conflict.   

\bigskip
All of the above word sequences have length at most $\ell^2 = O(w)$, and thus, each step takes constant time. Recall that we assumed all entries in $X$ were distinct. If not, we may have a write conflict at the end of step 3. To fix this, we can always double the length of the input word sequence $X$ and represent each entry using two words consisting of the entry and the position in the sequence, thus breaking ties. This operation can be performed by applying a scattered read on $X$ with the indices $\angle{0, \ell-1, 0, \ell-2, \ldots}$, that moves the  values of $X$ to the even indices of a word sequence of length $2\ell$, and then $\mid$'ing with the word sequence $\angle{\ell-1, 0, \ell-2, 0, \ldots}$. We can then simulate the above algorithms with constant factor slowdown by using $2w$ bits of precision for each component, as discussed in Section \ref{sec:extendedultrawords}. In summary, we have shown the following result.

\begin{lemma}\label{lem:prefixminsmall}
	Given a word sequence $X$ of length $\ell = O(\sqrt{w})$, we can compute the prefix minimum of $X$ in $O(1)$ time. 
\end{lemma}


As a corollary, we implement the above operation to compute prefix minima of blocks in a $b$-way word sequence $X$. More precisely, let $X=X_{s-1} \cdots X_0$ be a $b$-way word sequence and define the \emph{$b$-way prefix minimum} of $X$ to be 
\[ \pmini^b(X) = \pmini(X_{s-1})\cdot\pmini(X_{s-2})\cdots \pmini(X_0)\;.\] 

To compute $\pmini^b(X)$, where $X$ has length $\ell = O(w/b)$, consider each block as a word sequence of size $b$ to which we apply the above algorithm in parallel as follows. 
In Step 1, we create two word sequences $\widehat{X}^{b^2}$ and $\widetilde{X}^{b^2}$ that are the concatenation of all the word sequences $\widehat{X}$ and $\widetilde{X}$ of size $b^2$ corresponding to each block, and we apply all the operations outlined above in parallel for each subsequence of $b^2$ words. The concatenation can be done with a shuffled read on $X$ with the pre-computed word sequences
\begin{align*}    
    \widehat{A}^{b^2} &=\underbrace{\angle{\ell-1, \ldots, \ell - b}\cdots \angle{\ell-1, \ldots, \ell - b}}_{b}\cdots \underbrace{\angle{2b-1, \ldots, b}\cdots\angle{2b-1, \ldots, b}}_{b}\cdot \underbrace{\angle{b-1, \ldots, 0}\cdots\angle{b-1, \ldots, 0}}_{b}, \\
    \widetilde{A}^{b^2} &= \underbrace{\angle{\ell-1, \ldots, \ell-1} \cdots \angle{\ell-b, \ldots, \ell-b}}_{b}\cdots \underbrace{\angle{2b-1, \ldots, 2b-1}\cdots{\angle{b, \ldots, b}}}_{b}\cdot\underbrace{\angle{b-1,\ldots b-1}\cdots \angle{0, \ldots, 0}}_{b}.
\end{align*}

In Step 2, the only change needed is to use a mask $M^{b^2}$ which is the concatenation of $s$ copies of the original mask $M$. For step 3, the corresponding $P'$ is equal to $\widehat{A}^{b^2}$, and $P''$ is again a word sequence with unique values greater than or equal to $\ell$.

Since the length of $X$ is $\ell=O(w/b)$, when $b\leq \sqrt{w}$ the number of blocks is $O(w/b^2)$ and the length of the word sequences $\widehat{X}$ and $\widetilde{X}$ is $O(w)$. Since the length of the word sequences we work with is $O(w)$ we can do all the operations in constant time. Note that, if $b>\sqrt{w}$, the condition $\ell=O(w/b)$ implies that $\ell=O(\sqrt{w})$, and thus we have a constant number of blocks to process, which we can do with the previous algorithm instead (see Lemma \ref{lem:prefixminsmall}). We have the following result:
 
%
%

\begin{corollary}\label{cor:prefixminparallel}
Given a $b$-way word sequence $X$ of length $O(w/b)$ we can compute the $b$-way prefix minimum of $X$ in constant time. 
\end{corollary}

\subsection{Prefix Minima on General Word Sequences} 
We now show how to recursively apply Corollary~\ref{cor:prefixminparallel} to compute the prefix minima on a word sequence $X$ of length $\ell = O(w)$. Given the $b$-way prefix minimum of $X$, we show how to compute the $b^2$-way prefix minimum of $X$ in constant time. We then show how to apply this recursively to obtain our $O(\log \log \ell)$ algorithm. Let $X^b = X^b_{s-1} \cdots X^b_{0}$ be the $b$-way prefix minimum of $X$.  Our algorithm proceeds as follows (see Figure \ref{fig:pref_min_step_5_2}). 


\paragraph{Step 1: Compute Prefix Minima of Block Minima}
We compute the word sequence 
$$
B^b = \pmini^b(X^b_{s-1}\angle{b - 1} \cdots X^b_{0}\angle{b-1})
$$
containing the $b$-way prefix minimum of the minimum (leftmost) entries of the blocks in $X^b$. To do so, we first shuffle the leftmost entries of the blocks into a word sequence $Y$ and then compute $B^b = \pmini^b(Y)$ using Corollary~\ref{cor:prefixminparallel} (note that $Y$ has length $O(w/b)$ as required). 

\paragraph{Step 2: Propagate Prefix Minima of Block Minima.}

We construct a $b^2$-way word sequence $\widetilde{B}^{b^2}$ that contains, for each block $X^b$, the minimum of the previous blocks of $X^b$ in the corresponding block of $b^2$ words. 

First, we compute the word sequences
\begin{align*}
\widehat{B}^{b^2} &= \underbrace{B^b\angle{s-1} \cdots B^b\angle{s-1}}_b \cdots \underbrace{B\angle{0} \cdots B\angle{0}}_b, \\
M &= \underbrace{0\cdots 0}_{b^2-b} \cdot \underbrace{1^w\cdots 1^w}_{b} \cdots \underbrace{0\cdots 0}_{b^2-b} \cdot \underbrace{1^w\cdots 1^w}_{b},
\end{align*}
where $\widehat{B}^{b^2}$ contains $b$ copies of every entry in $B^b$, and $M$ contains a repeated pattern of $b$ words of $1$s followed by $b^2-b$ $0$s. Both word sequences have length $\ell$. We compute $\widehat{B}^{b^2}$ using a shuffled read operation, and $M$ can be pre-computed. After this, we shift left $\widehat{B}^b$ by $b$ words, and we calculate the logical `$\mid$' of $\widehat{B}^b$ and $M$ to produce the word sequence $\widetilde{B}^b$. Finally, we compute the componentwise minimum of $\widetilde{B}^b$ and $X^b$ to produce $X^{b^2}$. 

To see why this is correct, consider a sequence of $b$ blocks of the $b^2$-way prefix minimum of $X$, denoted $X^{b^2}_{ib}, X^{b^2}_{ib+1},\allowbreak \ldots,\allowbreak X^{b^2}_{ib+b-1}$ for any $0\leq i < \ceil{\ell/b^2}$. We can obtain the value of any $X^{b^2}_{ib+j}$ by the componentwise minimum of $X^b_{ib+j}$, and a block that contains $b$ copies of the minimum of $X_{ib+j-1}, \allowbreak \ldots \allowbreak X_{ib}$. This minimum is computed in the $(j-1)$th word of $B^b_i$, so we copy that word $b$ times in $\widetilde{B}^b$, and we shift it left by $b$ words to compute the componentwise minimum with $X^b_{ib+j}$. Note that for $j=0$, the value of $X^{b^2}_{ib}$ is the same as the value of $X^b_{ib}$, therefore for these blocks, we instead compare to the maximum word available ($1^w$) in order not to change their value. See Fig \ref{fig:pref_min_step_5_2} for a visualization.

\medskip
Each operation uses constant time, and hence we have the following result. 

\begin{lemma}\label{lem:prefixminrecursion}
	Let $X$ be a word sequence of length $\ell = O(w/b)$. Given a $b$-way prefix minimum of $X$, we can compute a $b^2$-way prefix minimum of $X$ in constant time. 
\end{lemma}
It now follows that we can compute the prefix minimum of a word of length $\ell$ by applying Lemma~\ref{lem:prefixminrecursion} for double exponentially increasing values of $b$ over $O(\log \log \ell)$ rounds, starting with a 2-way prefix minimum, that can be computed directly by Corollary \ref{cor:prefixminparallel} since $b=2$ is a constant. Hence, we have the following result. 

\begin{theorem}\label{thm:wordsequenceprefixmin}
	Given a word sequence $X$ of length $\ell = O(w)$, we can compute the prefix minimum of $X$ in $O(\log \log \log \ell)$ time. 
\end{theorem} 

Plugging in Theorem~\ref{thm:wordsequenceprefixmin} into our reduction of Theorem~\ref{thm:reduction}, we have shown the main result of Theorem~\ref{thm:main}. The space consumption follows from the observation that if we pad all inputs to the prefix minimum to length $\ell$, the pre-computed word sequences that are required for Theorem 6 are always the same. Therefore, we only need to pre-compute $O(\log \log n)$ word sequences of length $O(\log n)$ for the data structure. This gives an extra space of $O(\log n\cdot \log \log n)=o(n)$.

\begin{figure}[t]
    \centering
    \scalebox{1.0}{
\begin{tikzpicture}
	\begin{pgfonlayer}{nodelayer}
		\node [align=right] (0) at (-4.75, 2) {$X=\langle$};
		\node [align=right] (1) at (-4.75, 1.5) {$X^b=\langle$};
		\node [align=right] (2) at (-2.5, 0.5) {$Y=\langle$};
		\node [align=right] (3) at (-2.5, 0) {$B^b=\langle$};
		\node [align=right] (4) at (-4.75, -1.75) {$\widehat{B}^b=\langle$};
		\node [align=right] (5) at (-4.75, -3.75) {$X^{b^2}=\langle$};
		\node  (6) at (-4, 2) {12};
		\node  (7) at (-4, 1.5) {9};
		\node  (8) at (-1.75, -1.75) {2};
		\node  (9) at (-4, -3.75) {2};
		\node  (10) at (-3.5, 2) {9};
		\node  (11) at (-3, 2) {10};
		\node  (12) at (-2.5, 2) {11};
		\node  (13) at (-1.75, 2) {6};
		\node  (14) at (-1.25, 2) {2};
		\node  (15) at (-0.75, 2) {3};
		\node  (16) at (-0.25, 2) {5};
		\node  (17) at (0.5, 2) {5};
		\node  (18) at (1, 2) {9};
		\node  (19) at (1.5, 2) {10};
		\node  (20) at (2, 2) {8};
		\node  (21) at (2.75, 2) {8};
		\node  (22) at (3.25, 2) {4};
		\node  (23) at (3.75, 2) {5};
		\node  (24) at (4.25, 2) {6};
		\node  (25) at (4.75, 2) {$\rangle$};
		\node  (26) at (-3.5, 1.5) {9};
		\node  (27) at (-3, 1.5) {10};
		\node  (28) at (-2.5, 1.5) {11};
		\node  (29) at (-1.75, 1.5) {2};
		\node  (30) at (-1.25, 1.5) {2};
		\node  (31) at (-0.75, 1.5) {3};
		\node  (32) at (-0.25, 1.5) {5};
		\node  (33) at (0.5, 1.5) {5};
		\node  (34) at (1, 1.5) {8};
		\node  (35) at (1.5, 1.5) {8};
		\node  (36) at (2, 1.5) {8};
		\node  (37) at (2.75, 1.5) {4};
		\node  (38) at (3.25, 1.5) {4};
		\node  (39) at (3.75, 1.5) {5};
		\node  (40) at (4.25, 1.5) {6};
		\node  (41) at (4.75, 1.5) {$\rangle$};
		\node  (42) at (-1.25, 0.5) {2};
		\node  (43) at (-0.75, 0.5) {5};
		\node  (44) at (-1.75, 0.5) {9};
		\node  (45) at (-0.25, 0.5) {4};
		\node  (46) at (-1.75, 0) {2};
		\node  (47) at (-1.25, 0) {2};
		\node  (48) at (-0.75, 0) {4};
		\node  (49) at (-0.25, 0) {4};
		\node  (50) at (-1.25, -1.75) {2};
		\node  (51) at (-0.75, -1.75) {2};
		\node  (52) at (-0.25, -1.75) {2};
		\node  (53) at (-3.5, -3.75) {2};
		\node  (54) at (-3, -3.75) {2};
		\node  (55) at (-2.5, -3.75) {2};
		\node  (56) at (0.5, -1.75) {4};
		\node  (57) at (1, -1.75) {4};
		\node  (58) at (1.5, -1.75) {4};
		\node  (59) at (2, -1.75) {4};
		\node  (60) at (-1.75, -3.75) {2};
		\node  (61) at (-1.25, -3.75) {2};
		\node  (62) at (-0.75, -3.75) {3};
		\node  (63) at (-0.25, -3.75) {4};
		\node  (64) at (2.75, -1.75) {4};
		\node  (65) at (3.25, -1.75) {4};
		\node  (66) at (3.75, -1.75) {4};
		\node  (67) at (4.25, -1.75) {4};
		\node  (68) at (0.5, -3.75) {4};
		\node  (69) at (1, -3.75) {4};
		\node  (70) at (1.5, -3.75) {4};
		\node  (71) at (2, -3.75) {4};
		\node  (72) at (2.75, -3.75) {4};
		\node  (73) at (3.25, -3.75) {4};
		\node  (74) at (3.75, -3.75) {5};
		\node  (75) at (4.25, -3.75) {6};
		\node  (76) at (4.75, -3.75) {$\rangle$};
		\node  (77) at (2.75, -2.75) {$1^w$};
		\node  (78) at (3.25, -2.75) {$1^w$};
		\node  (79) at (3.75, -2.75) {$1^w$};
		\node  (80) at (4.25, -2.75) {$1^w$};
		\node  (81) at (4.75, -1.75) {$\rangle$};
		\node  (82) at (-4, 1.25) {};
		\node  (83) at (-1.75, 1.25) {};
		\node  (84) at (0.5, 1.25) {};
		\node  (85) at (2.75, 1.25) {};
		\node  (86) at (-1.75, 0.75) {};
		\node  (87) at (-1.25, 0.75) {};
		\node  (88) at (-0.75, 0.75) {};
		\node  (89) at (-0.25, 0.75) {};
		\node  (90) at (0.25, 0.5) {$\rangle$};
		\node  (91) at (0.25, 0) {$\rangle$};
		\node  (92) at (-0.25, -0.25) {};
		\node  (93) at (-1.75, -0.25) {};
		\node  (94) at (-1.25, -0.25) {};
		\node  (95) at (-0.75, -0.25) {};
		\node  (96) at (4.25, -1.5) {};
		\node  (97) at (3.75, -1.5) {};
		\node  (98) at (3.25, -1.5) {};
		\node  (99) at (2.75, -1.5) {};
		\node  (100) at (2, -1.5) {};
		\node  (101) at (1.5, -1.5) {};
		\node  (102) at (1, -1.5) {};
		\node  (103) at (0.5, -1.5) {};
		\node  (104) at (-0.25, -1.5) {};
		\node  (105) at (-0.75, -1.5) {};
		\node  (106) at (-1.25, -1.5) {};
		\node  (107) at (-1.75, -1.5) {};
		\node  (108) at (-4, -1.75) {2};
		\node  (109) at (-3.5, -1.75) {2};
		\node  (110) at (-3, -1.75) {2};
		\node  (111) at (-2.5, -1.75) {2};
		\node  (112) at (-2.5, -1.5) {};
		\node  (113) at (-3, -1.5) {};
		\node  (114) at (-3.5, -1.5) {};
		\node  (115) at (-4, -1.5) {};
		\node  (116) at (-4, -2.75) {0};
		\node  (117) at (-3.5, -2.75) {0};
		\node  (118) at (-3, -2.75) {0};
		\node  (119) at (-2.5, -2.75) {0};
		\node  (120) at (-1.75, -2.75) {0};
		\node  (121) at (-1.25, -2.75) {0};
		\node  (122) at (-0.75, -2.75) {0};
		\node  (123) at (-0.25, -2.75) {0};
		\node  (124) at (0.5, -2.75) {0};
		\node  (125) at (1, -2.75) {0};
		\node  (126) at (1.5, -2.75) {0};
		\node  (127) at (2, -2.75) {0};
		\node [align=right] (128) at (-4.75, -2.75) {$M=\langle$};
		\node  (129) at (4.75, -2.75) {$\rangle$};
		\node [align=right] (130) at (-4.75, -3.25) {$\widetilde{B}^{b^2}=\langle$};
		\node  (131) at (-4, -3.25) {2};
		\node  (132) at (-3.5, -3.25) {2};
		\node  (133) at (-3, -3.25) {2};
		\node  (134) at (-2.5, -3.25) {2};
		\node  (135) at (-1.75, -3.25) {4};
		\node  (136) at (-1.25, -3.25) {4};
		\node  (137) at (-0.75, -3.25) {4};
		\node  (138) at (-0.25, -3.25) {4};
		\node  (139) at (0.5, -3.25) {4};
		\node  (140) at (1, -3.25) {4};
		\node  (141) at (1.5, -3.25) {4};
		\node  (142) at (2, -3.25) {4};
		\node  (143) at (2.75, -3.25) {$1^w$};
		\node  (144) at (3.25, -3.25) {$1^w$};
		\node  (145) at (3.75, -3.25) {$1^w$};
		\node  (146) at (4.25, -3.25) {$1^w$};
		\node  (147) at (4.75, -3.25) {$\rangle$};
		\node [align=right] (148) at (-5.125, -2.25) {$\widehat{B}^b \ll b=\langle$};
		\node  (149) at (-4, -2.25) {2};
		\node  (150) at (-3.5, -2.25) {2};
		\node  (151) at (-3, -2.25) {2};
		\node  (152) at (-2.5, -2.25) {2};
		\node  (153) at (-1.75, -2.25) {4};
		\node  (154) at (-1.25, -2.25) {4};
		\node  (155) at (-0.75, -2.25) {4};
		\node  (156) at (-0.25, -2.25) {4};
		\node  (157) at (0.5, -2.25) {4};
		\node  (158) at (1, -2.25) {4};
		\node  (159) at (1.5, -2.25) {4};
		\node  (160) at (2, -2.25) {4};
		\node  (161) at (4.75, -2.25) {$\rangle$};
		\node  (162) at (2.75, -2.25) {0};
		\node  (163) at (3.25, -2.25) {0};
		\node  (164) at (3.75, -2.25) {0};
		\node  (165) at (4.25, -2.25) {0};
	\end{pgfonlayer}
	\begin{pgfonlayer}{edgelayer}
		\draw [style=Right] (82.center) to (86.center);
		\draw [style=Right] (83.center) to (87.center);
		\draw [style=Right] (84.center) to (88.center);
		\draw [style=Right] (85.center) to (89.center);
		\draw [style=Right] (92.center) to (96.center);
		\draw [style=Right] (92.center) to (97.center);
		\draw [style=Right] (92.center) to (98.center);
		\draw [style=Right] (92.center) to (99.center);
		\draw [style=Right] (95.center) to (100.center);
		\draw [style=Right] (95.center) to (101.center);
		\draw [style=Right] (95.center) to (102.center);
		\draw [style=Right] (95.center) to (103.center);
		\draw [style=Right] (94.center) to (104.center);
		\draw [style=Right] (94.center) to (105.center);
		\draw [style=Right] (94.center) to (106.center);
		\draw [style=Right] (94.center) to (107.center);
		\draw [style=Right] (93.center) to (112.center);
		\draw [style=Right] (93.center) to (113.center);
		\draw [style=Right] (93.center) to (114.center);
		\draw [style=Right] (93.center) to (115.center);
	\end{pgfonlayer}
\end{tikzpicture}

}
\caption{Illustration of Step 2 of Section 5.2: computing the $16$-way prefix minimum of a word sequence $X$, given the $4$-way prefix minimum of $X$ in $X^b$. Note that $\widehat{B}^b$ is shifted left by $b$ words.}
\label{fig:pref_min_step_5_2}
\end{figure}

\section{Conclusion and Open Problems}
We have shown how to build an $O(n)$ space data structure to solve Dynamic Range Minimum Queries in $O(\log \log \log n)$ time in the UWRAM model. We conclude with three natural open problems: 
\begin{itemize}
    \item Our results hold even in the restricted UWRAM model. We wonder whether or not this is optimal in this model or even in the stronger multiplication UWRAM.
    \item Some variations of the dynamic RMQ problem also consider dynamic updates to the underlying array (see e.g., Arge et al.~\cite{AFSN2013}). It appears interesting and non-trivial to extend our techniques to this scenario.
    \item Some of our techniques in the paper are likely practical given the current state of processor technology (e.g., AVX-512 instruction set). We wonder if a practical version of our techniques implemented on a modern processor can outperform current state-of-the-art implementations of dynamic range minimum queries.
\end{itemize}




\bibliographystyle{abbrv}
\bibliography{paper}

\end{document}